\documentclass[lettersize,journal]{IEEEtran}

\usepackage[bookmarks,colorlinks]{hyperref}
\hypersetup{colorlinks,citecolor= red,filecolor= blue,linkcolor= blue,urlcolor=blue}
\usepackage{cite}
\usepackage{amssymb,amsfonts, amsthm}
\usepackage{graphicx}
\usepackage{epstopdf}
\usepackage{textcomp}
\usepackage{xcolor}
\usepackage{algorithm}
\usepackage{algorithmicx}
\usepackage{algpseudocode}
\usepackage{verbatim}
\providecommand{\keywords}[1]{\textbf{\textit{Index terms---}} #1}
\usepackage{algorithm}
\usepackage{multirow}
\usepackage{tabularx}
\usepackage{lipsum}
\usepackage{epsfig}
\usepackage{amsmath}
\usepackage{caption}
\usepackage{subcaption}
\usepackage{array}
\newcolumntype{L}[1]{>{\raggedright\let\newline\\\arraybackslash\hspace{0pt}}m{#1}}
\newcolumntype{C}[1]{>{\centering\let\newline\\\arraybackslash\hspace{0pt}}m{#1}}
\newcolumntype{R}[1]{>{\raggedleft\let\newline\\\arraybackslash\hspace{0pt}}m{#1}}
\def\multiset#1#2{\ensuremath{\left(\kern-.3em\left(\genfrac{}{}{0pt}{}{#1}{#2}\right)\kern-.3em\right)}}

\newtheorem{theorem}{Theorem}

\newtheorem{lemma}{Lemma}
\def\BibTeX{{\rm B\kern-.05em{\sc i\kern-.025em b}\kern-.08em
    T\kern-.1667em\lower.7ex\hbox{E}\kern-.125emX}}
\begin{document}
\title{Transformer-assisted Parametric CSI Feedback for mmWave Massive MIMO Systems}

\author{
    Hyungyu~Ju,~\IEEEmembership{Member,~IEEE},
    Seokhyun~Jeong,~\IEEEmembership{Member,~IEEE},
    Seungnyun~Kim,~\IEEEmembership{Member,~IEEE},
    Byungju~Lee,~\IEEEmembership{Member,~IEEE}, and
    Byonghyo~Shim,~\IEEEmembership{Senior~Member,~IEEE}\\
   
    \thanks{
			Hyungyu~Ju, Seokhyun~Jeong, and Byonghyo~Shim are with
			the Institute of New Media and Communications, 
			Seoul National University, 
			Seoul 08826,
			Republic of Korea 
			(e-mail: {hkjoo@islab.snu.ac.kr, shjeong@islab.snu.ac.kr, bshim@snu.ac.kr}).
       		Seungnyun~Kim is with 
        	the Laboratory for Information and Decision Systems (LIDS),
        	Massachusetts Institute of Technology,
        	Cambridge, MA 02139
        	USA 
                (e-mail: {snkim94@mit.edu}).
                Byungju~Lee is with the Department of Information and Telecommunication Engineering, Incheon National University,
                Incheon 22012,
                Republic of Korea 
                (e-mail: {bjlee@inu.ac.kr}).  (\textit{Corresponding author: ByungJu~Lee.}) 
	}		
}

\maketitle


\begin{abstract}
As a key technology to meet the ever-increasing data rate demand in beyond 5G and 6G communications, millimeter-wave (mmWave) massive multiple-input multiple-output (MIMO) systems have gained much attention recently.
To make the most of mmWave massive MIMO systems, acquisition of accurate channel state information (CSI) at the base station (BS) is crucial.
However, this task is by no means easy due to the CSI feedback overhead induced by the large number of antennas.
In this paper, we propose a parametric CSI feedback technique for mmWave massive MIMO systems.
Key idea of the proposed technique is to compress the mmWave MIMO channel matrix into a few geometric channel parameters (e.g., angles, delays, and path gains).
Due to the limited scattering of mmWave signal, the number of channel parameters is much smaller than the number of antennas, thereby reducing the CSI feedback overhead significantly.
Moreover, by exploiting the deep learning (DL) technique for the channel parameter extraction and the MIMO channel reconstruction, we can effectively suppress the channel quantization error. 
From the numerical results, we demonstrate that the proposed technique outperforms the conventional CSI feedback techniques in terms of normalized mean square error (NMSE) and bit error rate (BER).
\end{abstract}


\section{Introduction}
\IEEEPARstart{R}{ecently}, to accommodate the rapidly growing data demand for beyond 5G and 6G networks, the study of the upper mid-band (i.e., 7-24$\,$GHz) was approved at the International Telecommunication Union (ITU) World Radiocommunications Conference-23 (WRC-23)~\cite{WRC23}.
By leveraging the abundant frequency resource in the higher frequency bands, millimeter wave (mmWave) communications can support future-oriented services such as extended reality (XR) devices, digital twin, and metaverse~\cite{rappaport2013millimeter,ji2018ultra}.
Since the transmit signals in mmWave bands suffer from severe attenuation due to the serious path loss and atmospheric absorption, beamforming techniques realized by the massive multiple-input multiple-output (MIMO) antenna arrays have been popularly used to compensate the loss~\cite{ahn2022towards}.
To fully enjoy benefits of beamforming technique, a base station (BS) should acquire an accurate channel state information (CSI)~\cite{kim2022parametric}.

In 5G New Radio (NR), the user equipment (UE) selects a directional beam from the Type I codebook and then feeds back the beam index when the line-of-sight (LoS) component is dominant~\cite{flordelis2018massive}.
In case there exist multiple strong scatterers, Type II codebook can be used to generate a linear combination of beams, each of which is mapped to the scatterer~\cite{guo2022overview}.
A well-known drawback of the conventional codebook-based feedback mechanism is the quantization error and significant feedback overhead being proportional to the number of antennas and subcarriers~\cite{lee2015antenna}.

The standard body in 5G-Advanced (specified in 3rd Generation Partnership Project (3GPP) Release 18~\cite{3gpp.38.843}) decided to use the artificial intelligence (AI)-aided CSI feedback mechanism since the AI-based channel compression is shown to be effective in improving feedback accuracy while accessing a low feedback overhead~\cite{guo2022ai, 9838669}.
In~\cite{wen2018deep, liu2020efficient}, AI-based CSI feedback techniques that learn the mapping function between the channel matrix and codeword using deep neural network (DNN) have been proposed. 
Among various DNN architectures, a convolutional neural network (CNN) has been popularly studied for its ability to extract the local and spatial characteristics of the 2D-uniform planar array (UPA) channels~\cite{liu2020efficient}.
A potential problem of the CNN-based technique is that it does not capture the correlation among spaced apart channel elements in time, space, and frequency domains. 
In other words, long-term correlated feature caused by the UE mobility, blockage, and MIMO antenna array cannot be effectively extracted in the CNN-based techniques.

Another issue of the conventional feedback mechanism, in particular for the current 5G standard, is that they do not consider the rapid channel variation in mmWave systems.
Due to the high UE mobility, channel coherence time of the mmWave systems is in general very short.
For example, in $28\,$GHz band, the channel coherence time of the moving UE at the speed of $30\,$km/h is around $1.2\,$ms, while the smallest period of the uplink reference signal (e.g., sounding reference signal (SRS) in 5G NR) is $2\,$ms \cite{3gpp.36.211}.\footnote{The channel coherence time is $T_c=\frac{1}{f_m}$ where $f_m = \frac{v f_c}{c}$ is the maximum Doppler spread, $v$ is the speed of UE, $f_c$ is the carrier frequency, and $c$ is the speed of light. By plugging $v=30\,$km/h and $f_c=28\,$GHz, we obtain $T_c \approx 1.2\,$ms.}
Note that the shortened channel coherence time results in a drastic variation of mmWave channel, leading to a significant mismatch between the estimated and actual channels during the data transmission. This phenomenon, often called $\textit{channel}$ $\textit{aging}$ $\textit{effect}$, is a serious problem yet there are few tractable solutions for the problem at hand.

An aim of this paper is to propose an AI-assisted CSI feedback mechanism that addresses aforementioned problems while achieving a reduction in feedback overhead.
The proposed technique, referred to as $\textit{comprehensive}$ $\textit{parametric}$ $\textit{CSI}$ $\textit{feedback}$ $\textit{based}$ $\textit{on}$ $\textit{Transformer}$ $\textit{(COMPaCT)}$, compresses the full channel matrix into a few geometric channel parameters (i.e., angle, delay, and path gain) and then performs the quantization. After receiving the quantized version of the channel parameters, BS recovers it to the original channel matrix via the deep learning (DL)-based decoder.
Since the mmWave channel is expressed as a function of the geometric parameters and these parameters are continuously changed based on the movement of UE, the geometric parameters of mmWave channels are strongly correlated in time.
In real-world scenarios, fluctuations in small-scale fading due to environmental changes (e.g., blockage, tunnel, or scatterer) cause a rapid variation in geometric channel parameters, weakening the temporal correlation among channel samples. 
Even in such cases, long-term spatial correlation caused by the large-scale fading components tends to be preserved.
For this reason, to properly exploit both temporally adjacent and distant channels is of great importance.

As a main ingredient to estimate the downlink channel parameters from the sequence of past channels, we exploit $\textit{Transformer}$, a state-of-the-art DL model.
A key feature of Transformer is the attention mechanism quantifying the long- and short-term correlations between input and output sequences~\cite{vaswani2017attention}.
By assigning relatively large attention weights to input data (i.e., past channels) that are highly related to the output values (i.e., downlink geometric channel parameters), we can effectively extract downlink channel parameters, which helps to mitigate the channel aging effect.

The main contributions of this paper are as follows.
\begin{itemize}
    \item We propose a Transformer-assisted channel estimation technique to estimate the sparse geometric channel parameters at the desired (future) time instant of downlink beamforming.
    Specifically, by exploiting both temporally adjacent and distant channels, the proposed COMPaCT can accurately estimate the downlink channel parameters, which ensures accurate beamforming even in the presence of mobility.

    \item We propose a feedback mechanism that delivers the quantized version of geometric channel parameters to the BS. Since the mmWave channel is expressed as channel parameters of LoS and a few non-LoS (NLoS) paths, we can significantly reduce the dimension of a channel vector to be quantized.
    
    \item We present a feedback bit allocation strategy to minimize the channel quantization error caused by the mismatch between the desired and quantized channel parameters. 
    Since each channel parameter induces a different level of channel quantization error, we allocate feedback bits differently.
    In fact, the normalized mean square error (NMSE) of angle or delay is much higher than that of path gain when the same number of feedback bits is used , so we prioritize angle and delay over path gain in the feedback bit allocation.
    
    \item Through numerical results, we show that the proposed COMPaCT achieves a considerable gain over the conventional CSI feedback and channel acquisition schemes in terms of NMSE. In particular, the proposed COMPaCT achieves more than $3.5\,$dB and $2\,$dB NMSE gain over the conventional CNN and the reciprocity-based feedback mechanisms in practical mmWave scenarios. 

\end{itemize}

The rest of this paper is organized as follows. In Section II, we discuss the mmWave MIMO system and briefly explain conventional CSI feedback and channel acquisition techniques. In Section III, we present COMPaCT and a detailed description of the model. In Section IV, we propose the parameter selective bit allocation method for the quantization process. In Section V, we demonstrate experimental results to validate the efficacy of the proposed technique and then conclude the paper in Section VI.

\textit{Notations}: Upper and lower case symbols are used to denote matrices and vectors, respectively.
The superscript $\left(\cdot\right)^{\text{T}}$ and $\left(\cdot\right)^{\text{H}}$ denote the transpose and the Hermitian transpose, respectively.
$\mathbf{X}_1\otimes \mathbf{X}_2$ and $\mathbf{X}_1\circ \mathbf{X}_2$ denote the Kronecker product and Hadamard product of $\mathbf{X}_1$ and $\mathbf{X}_2$, respectively. 
$||\mathbf{x}||_2$ and $\Vert \mathbf{X} \Vert_F$ denote the Euclidean norm of a vector $\mathbf{x}$ and the Frobenius norm of the matrix $\mathbf{X}$, respectively. 
The $m$-by-$m$ identity matrix is denoted by $\mathbf{I}_m$.
Also, diag($\mathbf{X}$) denotes a block diagonal matrix whose diagonal elements are $\mathbf{X}$.
$\Re(x)$ and $\Im(x)$ are the real and imaginary part of $x$, respectively.
In addition, $\mathbf{0}_K$ and $\mathbf{1}_K$ denote $K\times 1$ zero vector and one vector, respectively.

\section{mmWave Massive MIMO System Model}
In this section, we briefly explain the downlink mmWave system model and then review the conventional CSI feedback and channel acquisition techniques. 

\subsection{Downlink mmWave System Model}
We consider the frequency division duplexing (FDD) MIMO-OFDM (orthogonal frequency division multiplexing) downlink systems where BS equipped with $N_t$ transmit antennas serves the UE equipped with $N_r$ antennas.
Specifically, $N_f$ subcarriers and $T$ time frames are used for the downlink pilot transmission.
The carrier frequency is $f_c$ and the bandwidth is $B$.
In our work, we use $\mathcal{S} = \{1,\cdots,N_f\}$ and $\mathcal{T}=\{1,\cdots,T\}$ to denote the sets of indices of pilot subcarriers and time frames.
Under the assumption that $N_t$ $\gg$ $N_r$, we set $N_r = 1$ for simplicity.
In a time-varying channel, the received signal $y_{t}[s]\in\mathbb{C}$ of UE with $s$-th subcarrier at $t$-th time frame is
\begin{equation}
    y_t[s] = \mathbf{h}[s]^{\mathrm{H}} \mathbf{r}_t x_t[s] + n_t[s], \, \quad \forall s \in \mathcal{S}, \, t \in \mathcal{T}
\end{equation}
where $\mathbf{h}[s] \in \mathbb{C}^{N_t \times 1}$ is the downlink channel vector in the frequency domain, $\mathbf{r}_t\in \mathbb{C}^{N_t\times 1}$ is the beamforming vector based on the downlink CSI, $x_t[s] \in \mathbb{C}$ is the downlink pilot symbol, and $n_t[s]\sim\mathcal{CN}(0,\sigma_{n}^{2})$ is the additive Gaussian noise of $s$-th subcarrier.

In this work, we use the geometric multipath channel model where the time domain channel response vector at $t$-th time frame $\mathbf{h}(t)$ is expressed as
\begin{equation}
    \mathbf{h}(t) = \sum_{l=1}^{L} \alpha_{l}(t) \delta(\tau-\tau_l(t))\mathbf{a}_{t}(\theta_l(t)), \,\quad \forall t \in \mathcal{T} \label{channel_response}
\end{equation}
where $L$ is the number of effective propagation paths, $\theta_l\sim \text{Unif}[0,2\pi)$ is the angle of departure (AoD), $\tau_l\sim \text{Unif}[0,\tau_{\text{max}}]$ is the time delay, and $\alpha_l = \beta_l e^{j\phi_l}$ is the complex gain consisting of the path loss $\beta_l\sim \text{Unif}[0,\beta_{\text{max}}]$ and the phase $\phi_l\sim \text{Unif}[0,2\pi)$ of the $l$-th path, respectively. 
Note that $\tau_{\text{max}}$ and $\beta_{\text{max}}$ are the maximum time delay and path loss, respectively.
Also, $\mathbf{a}_{t}(\theta_l(t))\in \mathbb{C}^{N_t \times 1}$ is the array steering vector of BS, which can be expressed as 
\begin{align}
\mathbf{a}_{t}(\theta_l(t)) =  \Big[1\, e^{-j\frac{2\pi d \sin{\theta_l(t)}}{\lambda}} \dotsc e^{-j(N_t-1)\frac{2\pi d \sin{\theta_l(t)}}{\lambda}}\Big]^{\mathrm{T}}
\end{align}
where $\lambda$ is the signal wavelength and $d$ is the antenna spacing.
Applying the discrete Fourier transform (DFT), the frequency domain channel vector $\mathbf{h}[s]\in \mathbb{C}^{N_t \times 1}$ becomes
\begin{align}
    \mathbf{h}[s] = \sum_{l=1}^{L} \beta_l(t) e^{j\phi_l(t)}e^{-j2\pi f_s \tau_l(t) }\mathbf{a}_{t}(\theta_l(t)), \, \quad \forall s \in \mathcal{S}  \label{freq_channel}
\end{align}
where $f_s = f_c -\frac{B}{2}+\frac{B}{N_f}(s-1)$ is the $s$-th subcarrier frequency.
By concatenating $\mathbf{h}[1],\cdots,\mathbf{h}[N_{f}]$, we obtain the CSI matrix of the frequency domain $\mathbf{H}\in\mathbb{C}^{N_{f}\times N_t}$ given by
\begin{align}
    \mathbf{H} &=\big[\mathbf{h}[1]\cdots\mathbf{h}[N_{f}]\big]^{\mathrm{H}}  \\ 
    &= \sum_{l=1}^{L} \beta_l(t) e^{-j\phi_l(t)}\mathbf{a}_f(\tau_l(t))\mathbf{a}^{\mathrm{H}}_{t}(\theta_l(t))
\label{channel-m}
\end{align}
where $\mathbf{a}_f(\tau_l(t))\in\mathbb{C}^{N_{f} \times 1}$ is the phase shift vector of the OFDM subcarriers defined as~\cite{han2020deep}
\begin{equation}
    \mathbf{a}_f(\tau_l(t)) =  \big[e^{j2\pi \tau_l(t) f_1}\, e^{j2\pi \tau_l(t) f_2}\cdots\, e^{j2\pi \tau_l(t) f_{N_f}}\big]^{\mathrm{T}}.
\end{equation}
In a matrix form, $\mathbf{H}$ can be expressed by the geometric channel parameters:
\begin{align}
    &\,\,\mathbf{H}(\boldsymbol{\theta}(t), \boldsymbol{\tau}(t), \boldsymbol{\beta}(t), \boldsymbol{\phi}(t)) \\ &= \mathbf{A}_f(\boldsymbol{\tau}(t))\mathrm{diag}(\boldsymbol{\beta}(t))\mathrm{diag}(e^{-j\boldsymbol{\phi}(t)})\mathbf{A}^{\mathrm{H}}_{t}(\boldsymbol{\theta}(t))
    \label{parameterize}
\end{align}
where $\mathbf{A}_{t}(\boldsymbol{\theta}(t))\in\mathbb{C}^{N_{t}\times L}$ and $\mathbf{A}_{f}(\boldsymbol{\tau}(t))\in\mathbb{C}^{N_{f}\times L}$ are the array steering matrices given by
\begin{align}
    \mathbf{A}_{t}(\boldsymbol{\theta}(t)) =& [\mathbf{a}_{t}(\theta_1(t)),\cdots, \mathbf{a}_{t}(\theta_L(t))]\\
    \mathbf{A}_{f}(\boldsymbol{\tau}(t)) =& [\mathbf{a}_f(\tau_1(t)),\cdots, \mathbf{a}_f(\tau_L(t))]
\end{align}
and $\boldsymbol{\theta}(t)=[\theta_{1}(t),\cdots,\theta_{L}(t)]$, $\boldsymbol{\tau}(t)=[\tau_{1}(t),\cdots,\tau_{L}(t)]$,  $\boldsymbol{\beta}(t)=[\beta_{1}(t),\cdots,\beta_{L}(t)]$, and $\boldsymbol{\phi}(t)=[\phi_{1}(t),\cdots,\phi_{L}(t)]$.
Since the mmWave channel matrix $\mathbf{H}$ can be expressed as a function of AoDs $\boldsymbol{\theta}$, delays $\boldsymbol{\tau}$, path losses $\boldsymbol{\beta}$, and phases $\boldsymbol{\phi}$, variation of these parameters due to UE mobility will lead to a fast change in channel (i.e., $\mathbf{h}(t+\Delta t) \neq \mathbf{h}(t)$).
Thus, without a proper consideration of this channel aging effect, there will be a significant mismatch between the original and reconstructed channels, especially when used for the downlink beamforming. 

\subsection{Conventional mmWave MIMO CSI Feedback and Downlink Channel Acquisition}

\subsubsection{MIMO CSI Feedback}

In the codebook-based CSI feedback method such as the random vector quantization~\cite{love2008overview}, UE estimates the downlink channel vector $\mathbf{h}$ using the downlink pilot signal and then quantizes the normalized channel vector $\Bar{\mathbf{h}}=\frac{\mathbf{h}}{||\mathbf{h}||}$. This is done by picking the quantized vector (codeword) $\mathbf{c}_{\hat{i}}$ from a pre-defined $B$-bit quantization codebook $\mathcal{C}=\{\mathbf{c}_1,\cdots,\mathbf{c}_{2^B}\}$  closest to the channel direction:
\begin{align}
    \mathbf{c}_{\hat{i}} = \arg \max_{\mathbf{c}_i\in \mathcal{C}} |\Bar{\mathbf{h}}^{\mathrm{H}}\mathbf{c}_i|^2.
\end{align}
Then, the UE feeds back the index of the selected codeword $\hat{i}$ to the BS via the uplink control channel (e.g., PUCCH~\cite{3gpp.36.212}).
It has been shown that the number of feedback bits $B$ should be proportional to the channel dimension $N_f N_t$ and the SNR (in decibels) to control quantization distortion properly~\cite{kim2020downlink}
\begin{align}
    B \approx \frac{N_f N_t - 1}{3} \times \text{SNR}.
\end{align}
In the mmWave massive MIMO systems, substantial feedback resource is needed due to a large number of antennas and subcarriers. 
Lack of feedback resources will result in a degradation of CSI reconstruction performance, causing a significant drop in the beamforming gain. 

To reduce huge feedback overhead in massive MIMO systems, CSI extrapolation techniques have been proposed~\cite{lin2021deep, zhang2021deep}.
Further, DL-based CSI feedback techniques using CNN or long short-term memory (LSTM) that compress large-dimensional channels into an implicit codeword have been proposed~\cite{wen2018deep, wang2018deep}.
In conventional DL-based approaches, UE compresses the input channel matrix $\mathbf{H}$ into the codeword $\mathbf{c}$ via the encoder $f_{\text{enc}}$ (i.e., $\mathbf{c} = f_{\text{enc}}(\mathbf{H}; \boldsymbol{\delta}) $, where $\boldsymbol{\delta}$ is the network parameters of the encoder) \cite{kim2023towards}.
After BS receives $\mathbf{c}$ from UE, the codeword $\mathbf{c}$ passes through the decoder $f_{\text{dec}}$ to generate the reconstructed CSI $\mathbf{\hat{H}}$ (i.e., $\mathbf{\hat{H}} = f_{\text{dec}}(\mathbf{c}; \boldsymbol{\gamma}) $, where $\boldsymbol{\gamma}$ is the network parameters of the decoder).  
The overall end-to-end CSI feedback process can be expressed as
\begin{align}
\label{Ha_reconstruction}
    \hat{\mathbf{H}} = f_{\text{dec}}(f_{\text{enc}}(\mathbf{H}; \boldsymbol{\delta}); \boldsymbol{\gamma}).
\end{align}
There are two main issues for the conventional DL-based CSI feedback techniques.
First, conventional DL-based techniques suffer from the channel aging effect since they do not consider the temporal channel variation caused by UE mobility.
Second, conventional LSTM and CNN-based approaches are not so efficient in capturing the temporally or spatially distant correlation of mmWave channels, since these approaches are specialized in extracting the temporally/spatially adjacent correlation of the channel.\footnote{The number of CNN's kernel window is $3\sim5$ and the LSTM suffers from vanishing gradient problem~\cite{shi2015convolutional}.}

\subsubsection{Downlink Channel Acquisition}
\begin{figure*}[t]
\begin{center}
  \begin{subfigure}{0.35\textwidth}
    \includegraphics[width=\linewidth]{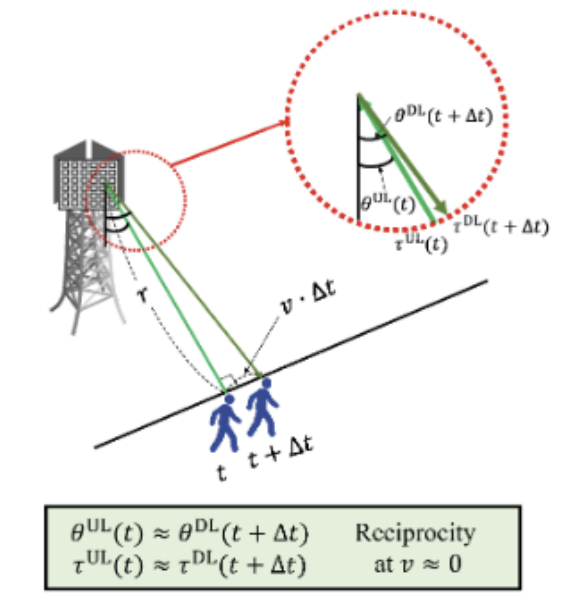}
    \caption{Stationary scenario} \label{fig:1a}
  \end{subfigure}%
  \hspace{1.0cm}
  \begin{subfigure}{0.35\textwidth}
    \includegraphics[width=\linewidth]{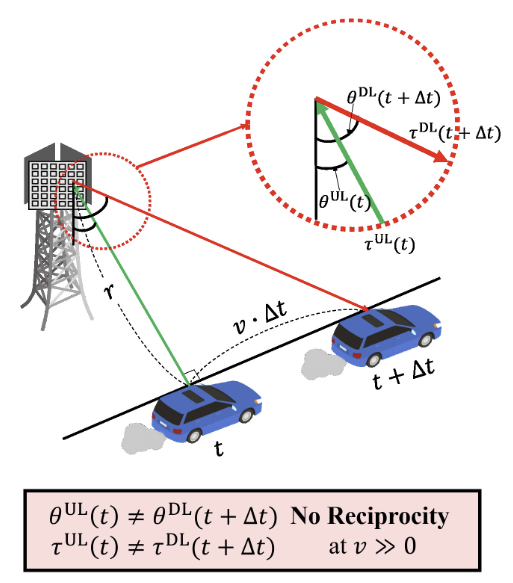}
    \caption{Non-stationary scenario} \label{fig:1b}
  \end{subfigure}
  \caption{mmWave Massive MIMO system with moving UE; The angle-delay reciprocity holds in stationary (quasi-static) but not in non-stationary scenarios.} \label{system}
\end{center}
\vspace{-0.2cm}
\end{figure*}
The conventional downlink channel acquisition technique relies on the angular-delay reciprocity between the uplink and downlink paths~\cite{han2019tracking}, since there is essentially no difference between the path angle and time delay for the uplink and the downlink.
That is,
\begin{align}
    \theta^{\text{DL}}_l(t+\Delta t) &\approx \theta^{\text{UL}}_l(t) \label{angle_c}\\
    \tau_l^{\text{DL}}(t+\Delta t) &\approx \tau^{\text{UL}}_l(t) \label{delay_c}
\end{align}
where $t+\Delta t$ is the time unit for downlink beamforming.
The frequency-independent parameters (i.e., AoDs $\boldsymbol{\theta}$ and delays $\boldsymbol{\tau}$) are extracted from the uplink pilot signal, and the complex gain can be estimated by exploiting the angular-delay reciprocity.

Although the assumptions in (\ref{angle_c}) and (\ref{delay_c}) hold true in a stationary scenario, they might not work well for the non-stationary scenario, especially in the presence of UE mobility (see Fig.~\ref{system}). 
For example, when the distance between BS and UE is $r=10\,\text{m}$ and the UE moves at a speed of $v =72$ km/h during $\Delta t = 10\,\text{ms}$ (specified in 3GPP TS 38.211~\cite{3gpp.38.211}), angle variation $\Delta \theta$ and delay variation $\Delta \tau$ are approximated to
\begin{align}
    \Delta \theta &= \arctan \Big(\frac{v\cdot \Delta t}{r}\Big) \approx 1.143^{\circ} \\
    \Delta \tau &= \frac{\sqrt{r^2+(v\cdot \Delta t)^2}-r}{c} = 5\times 10^{-12}
\end{align}
where $c$ is the speed of light. 
\begin{figure}[t]
    \centering 
    \includegraphics[width=1\linewidth]{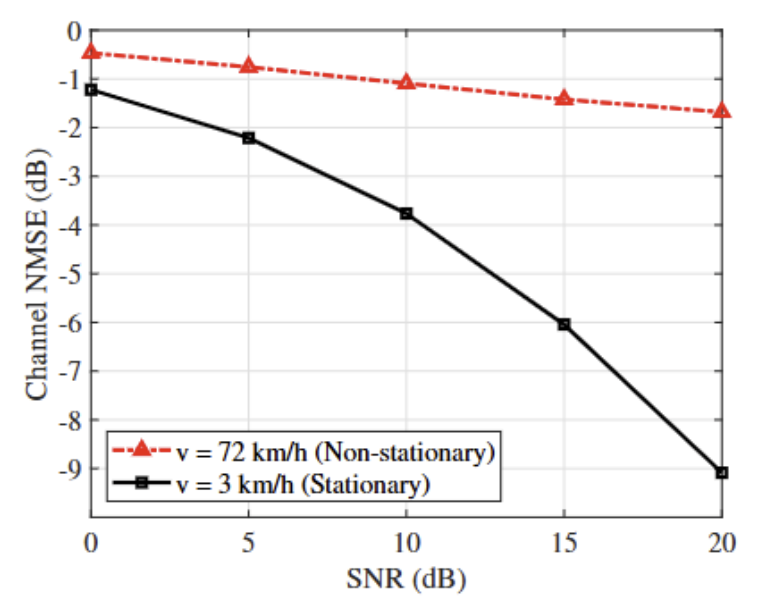}
    \caption{Channel NMSE according to UE mobility.}
    \label{result_sys}
\end{figure}
To evaluate the impact of these variations, we plot the channel NMSE between the true channel and the reconstructed channel: 
\begin{align}
    \text{NMSE} = \mathbb{E}\left[\frac{\Vert \hat{\mathbf{H}} -\mathbf{H}\Vert_{\mathrm{F}}^2}{\Vert \mathbf{H}\Vert_{\mathrm{F}}^2}\right].
\end{align}

As shown in Fig.~\ref{result_sys}, variations of angle and delay in the non-stationary scenario cause an additional $3\sim4\,$dB loss in the channel NMSE compared to the stationary scenario \cite{wu2016survey}.
Particularly, the angle mismatch $\Delta \theta$ between the pre-defined beam direction and real one might result in a significant beam misalignment, causing a severe loss in beamforming gain.\footnote{For example, within the 28 GHz bands, the half power beam width, representing the angle range where the relative power exceeds 50$\%$ of the peak power of the main beam, is around 10 degrees \cite{rappaport2013millimeter}.}

\section{Transformer-based Parametric CSI Feedback}

In this section, we propose a DL-based channel feedback mechanism to deal with the channel aging effect while achieving a reduction in the feedback overhead.
Exploiting the property that the channel matrix $\mathbf{H}$ can be parameterized by AoDs, time delays, path losses, and phases, the proposed COMPaCT feeds back the quantized version of the parametric CSI $\mathbf{P}\in \mathbb R^{L \times 4}$: 
\begin{align}
    \mathbf{P} = \{\boldsymbol{\theta}, \boldsymbol{\tau}, \boldsymbol{\beta}, \boldsymbol{\phi}\}.
\end{align}
Since the number of effective propagation paths $L$ is much smaller than the total number of antennas $N_t$ (e.g., $L=2\sim 5$ and $N_t=32\sim 256$), compression of parametric CSI is far more efficient than the codebook-based channel feedback.
For example, if the numbers of antennas and subcarriers are $64$ and $1024$, then the original channel dimension for the feedback is $N_f N_t=2^{16}$ but the dimension of $\mathbf{P}$ is $4L\approx 10\sim25$.\footnote{Note that the compression ratio of the proposed COMPaCT is $\frac{4L}{2 N_f N_t}$. When $L = 4$, $N_f = 1024$, and $N_t = 64$, the compression ratio of COMPaCT is $\frac{4\cdot 4}{2\cdot1024\cdot 64}=\frac{1}{9192}$ compared to transmitting the full channel matrix.}

\begin{figure*}[t] 
\centering 
\includegraphics[width=0.9\linewidth]{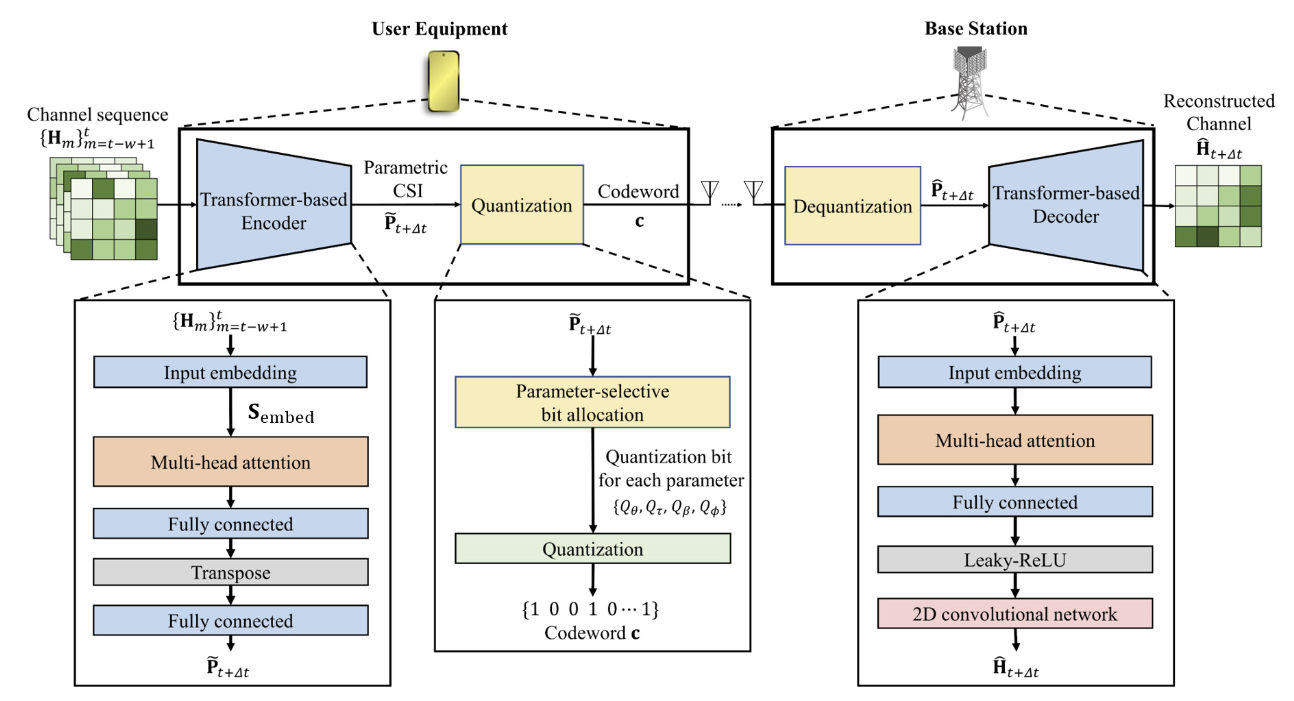}
\vspace{-0.3cm}
\caption{Overall process of the proposed COMPaCT.}
\label{overall}
\end{figure*}
Another key ingredient of COMPaCT is the Transformer-based framework that estimates the parametric CSI of the downlink channel $\mathbf{P}_{t+\Delta t}$ from the $w$-step channel sequence $\{\mathbf{H}_m\}_{m=t-w+1}^{t}$.
As mentioned, Transformer is effective in extracting temporally and spatially correlated features of mmWave massive MIMO channels.
By measuring the long- and short-term correlations between the previous channels using the attention block, and then assigning relatively large attention weights to the (previous) channels that are strongly correlated with the desired channel time $t+\Delta t$, the proposed COMPaCT can accurately estimate downlink parametric CSI for the desired time.

\subsection{Overall Process of COMPaCT}
The proposed COMPaCT consists of two main components: 1) Transformer-based encoder that generates the geometric channel parameters from sequence of historical channels and 2) Transformer-based decoder to reconstruct the original CSI from the quantized parameters.
First, Transformer-based encoder generates the parametric CSI $\widetilde{\mathbf{P}}_{t+\Delta t}$ $\in \mathbb{R}^{L\times 4}$ as
\begin{align}
    \widetilde{\mathbf{P}}_{t+\Delta t} = f_{\text{enc}}(\{\mathbf{H}_m\}_{m=t-w+1}^{t}, \boldsymbol{\delta})
    \label{1.p}
\end{align}
where $f_{\text{enc}}$ and $\boldsymbol{\delta}$ are the channel compression module and network parameters of the proposed COMPaCT encoder, respectively.
Second, the parametric CSI $\widetilde{\mathbf{P}}_{t+\Delta t}$ is quantized into the codeword $\mathbf{c}$ by the quantization module in UE as
\begin{align}
    \mathbf{c} = f_{\text{qnt}}(\widetilde{\mathbf{P}}_{t+\Delta t})
\end{align}
where $f_{\text{qnt}}$ is the quantization function of proposed COMPaCT.
Then, UE feeds back the quantized codeword $\mathbf{c}$ to BS via the feedback link. Using this information, the decoder at the BS restores the AoDs $\boldsymbol{\hat{\theta}}=[\hat{\theta}_{1},\cdots,\hat\theta_{L}]$, the time delays $\boldsymbol{\hat{\tau}}=[\hat{\tau}_{1},\cdots,\hat\tau_{L}]$, the path losses $\boldsymbol{\hat{\beta}}=[\hat{\beta}_{1},\cdots,\hat\beta_{L}]$, and the phases $\boldsymbol{\hat{\phi}}=[\hat{\phi}_{1},\cdots,\hat\phi_{L}]$ from the codeword $\mathbf{c}$,
\begin{align}
    \mathbf{\hat{P}}_{t+\Delta t}=&\{\hat{\boldsymbol{\theta}}, \hat{\boldsymbol{\tau}}, \hat{\boldsymbol{\beta}}, \hat{\boldsymbol{\phi}}\} \\
    =& f_{\text{de-qnt}}(\mathbf{c})
\end{align}
where $f_{\text{de-qnt}}$ is the mapping function of the de-quantization process.
Finally, we reconstruct the downlink channel $\mathbf{H}_{t+\Delta t}$ using the Transformer-based decoder as
\begin{align}
    \mathbf{\hat{H}}_{t+\Delta t}
    =f_{\text{dec}}(\mathbf{\hat{P}}_{t+\Delta t}, \boldsymbol{\gamma})
\end{align}
where $f_{\text{dec}}$ and $\boldsymbol{\gamma}$ are the channel reconstruction module and the network parameters of the Transformer-based decoder.

\subsection{Transformer-based Encoder with Parametric CSI Compression}
As mentioned, Transformer-based encoder outputs the parametric CSI $\widetilde{\mathbf{P}}_{t+\Delta t}$ from the sequence of historical channels (see Fig.~\ref{proposed1}).
\begin{figure*}[t] 
\centering 
\includegraphics[width=0.7\linewidth]{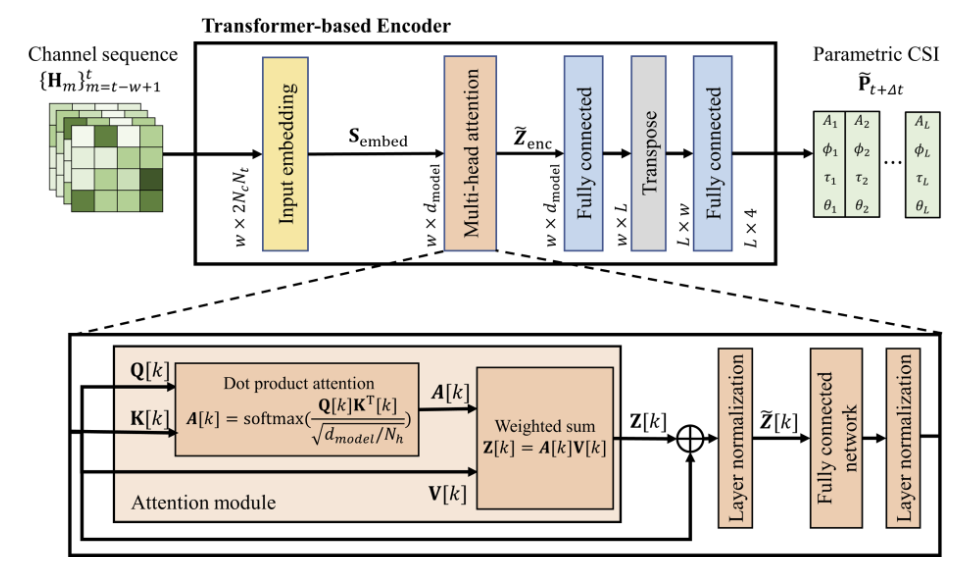}
\caption{Transformer-based encoder of the proposed COMPaCT.}
\label{proposed1}
\end{figure*}
As an input to the encoder, sequence of estimated channel matrices $\{\mathbf{H}_{t-w+1},$ $\cdots,$ $\mathbf{H}_{t}\} \in \mathbb{C}^{w\times N_f N_t}$ is used.
The input matrix $\widetilde{\mathbf{H}}_{\text{input}} \in \mathbb{R}^{w\times 2N_f N_t}$ of the Transformer-based encoder is
\begin{equation}
    \widetilde{\mathbf{H}}_{\text{input}} =
    \begin{bmatrix}
       \Re(\text{vec}(\mathbf{H}_{t-w+1}))&\cdots&\Re(\text{vec}(\mathbf{H}_{t})) \\
       \Im(\text{vec}(\mathbf{H}_{t-w+1}))&\cdots&\Im(\text{vec}(\mathbf{H}_{t})) \\
    \end{bmatrix}^{\text{T}}.
    \label{vectorize}
\end{equation}
We obtain $\mathbf{S}_{\text{embed}}$ from $\widetilde{\mathbf{H}}_{\text{input}}$ in the input embedding process, which will be delivered to the multi-head self-attention layer.
In the input embedding process, we first generate the angle-delay domain channel matrix using a 2D discrete Fourier transform (DFT) and retain only the first $\tilde{N}_c = 32$ rows.
This is because path components with long delays are significantly attenuated according to Friis' law, allowing us to disregard them~\cite{wen2018deep}. 
We then take the vectorized angle-delay domain channel matrix as the input for two fully connected networks to obtain $\mathbf{S}_{\text{embed}}$. Multi-head attention structure is in particular useful for our work since each head takes care of individual channel parameter (i.e., angle, delay, and path gain).

When the number of attention heads is $N_h$, the query $\mathbf{Q}[k] \in \mathbb{R}^{w\times d_{\text{model}}}$, key $\mathbf{K}[k] \in \mathbb{R}^{w\times d_{\text{model}}}$, and value $\mathbf{V}[k] \in \mathbb{R}^{w\times d_{\text{model}}}$ at $k$-th attention head layer ($k = 1, 2, \cdots, N_h$) are constructed from $\mathbf{S}_{\text{embed}}$ as
\begin{align}
    \mathbf{Q}[k] &= \mathbf{S}_{\text{embed}} \, \mathbf{W}^Q[k]\\ 
    \mathbf{K}[k] &= \mathbf{S}_{\text{embed}} \, \mathbf{W}^K[k]\\
    \mathbf{V}[k] &= \mathbf{S}_{\text{embed}} \, \mathbf{W}^V[k]
\end{align}
where $d_{\text{model}}$ is the embedding dimension and $\mathbf{W}^Q[k]\in\mathbb{R}^{d_{\text{model}} \times (d_{\text{model}}/N_h)}$, $\mathbf{W}^K[k]\in \mathbb{R}^{d_{\text{model}} \times (d_{\text{model}}/N_h)}$, and $\mathbf{W}^V[k]\in  \mathbb{R}^{d_{\text{model}} \times (d_{\text{model}}/N_h)}$ are the trainable matrices with linear transformation, respectively.
Since the query $\mathbf{Q}$ and key $\mathbf{K}$ contain features of the input channel sequence $\widetilde{\mathbf{H}}_{\text{input}}$, we obtain the attention map $\boldsymbol{A}[k] \in \mathbb{R}^{w\times w}$ as
\begin{align}
    \boldsymbol{A}[k] = f_{\text{softmax}}(\frac{\mathbf{Q}[k] \mathbf{K}^{\text{T}}[k]}{\sqrt{d_{\text{model}}/N_h}})
    \label{att-s}
\end{align}
where $f_{\text{softmax}}$ is the row-wise softmax function defined as $\big[f_{\text{softmax}}(\mathbf{X})\big]_{i,j} = \frac{e^{[\mathbf{X}]_{i,j}}}{\sum_j e^{[\mathbf{X}]_{i,j}}}$.
Each column in the attention map represents a probability vector whose elements are all non-negative and add up to one.

\begin{figure*}[t] 
\centering 
\includegraphics[width=0.85\linewidth]{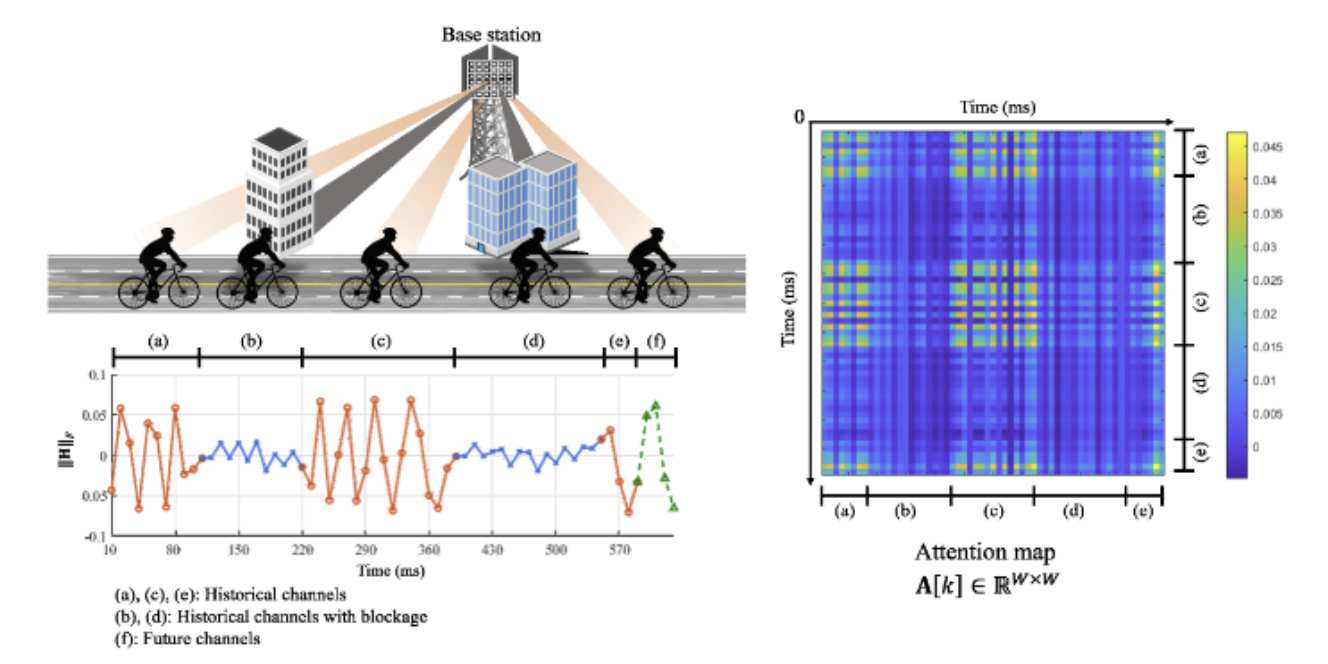}
\caption{An example of the practical scenario in mmWave massive MIMO systems; the LoS dominant parts of the channel sequence (i.e., (a), (c), and (e)) can be exploited to acquire the future channels (f) from the attention map.}
\label{channel}
\end{figure*} 

To demonstrate the effect of the attention mechanism in capturing the correlation structure of the past channels in $\{\mathbf{H}_m\}_{m=t-w+1}^{t}$, we plot the attention map of the COMPaCT encoder in Fig.~\ref{channel}.
From the attention map, one can see that the attention block automatically assigns more weights to parts of input channels (i.e., (a), (c), and (e)), which are most relevant to future channels (i.e., (f)).

Then, the feature map of the channel sequence $\boldsymbol{A}[k]$ is multiplied by the values $\mathbf{V}[k]$ on all heads $\mathbf{Z}[k] = \boldsymbol{A}[k] \mathbf{V}[k] \in \mathbb{R}^{w\times (d_{\text{model}}/N_h)}$.
By concatenating $\mathbf{Z}[k]$, we obtain the final output $\mathbf{Z}_{\text{enc}}$ of the multi-head self-attention layer:
\begin{align}
    \mathbf{Z}_{\text{enc}} = \left[\mathbf{Z}[1], \mathbf{Z}[2], \cdots, \mathbf{Z}[N_h] \right] \in \mathbb{R}^{w\times d_{\text{model}}}.
\end{align}
After this, we perform the layer normalization.
The output matrix $\widetilde{\mathbf{Z}}_{\text{enc}}$ after the layer normalization is
\begin{align}
    \widetilde{\mathbf{Z}}_{\text{enc}} = \kappa \left(\frac{\mathbf{Z}_{\text{enc}} - \mu}{\sqrt{\sigma^2+\epsilon}}\right)+\eta
\label{norm}
\end{align}
where $\mu$ and $\sigma^2$ are mean and variance, respectively. $\kappa$ and $\eta$ are learnable scaling and shifting parameters, respectively, and $\epsilon$ is a small constant for numerical stability.
Note that variation of $\mathbf{H}$ would be large since the mmWave environments suffer from severe path loss. A drastic change in gradient during the weight update process might affect the learning stability.
So, using the layer normalization, we can mitigate the fluctuation in the weight update (see Equation~\eqref{norm}).

After the layer normalization process, we use the fully connected layers to obtain the parametric CSI $\widetilde{\mathbf{P}}_{t+\Delta t}\in \mathbb{R}^{L\times 4}$, which can be expressed as
\begin{align}
    \widetilde{\mathbf{P}}_{t+\Delta t} = f_{\text{fc}}(\widetilde{\mathbf{Z}}_{\text{enc}}, \mathbf{W}_c)
\end{align}
where $f_{\text{fc}}$ is the layers that consist of fully connected networks and $\mathbf{W}_c$ is the set of network parameters in the fully connected layers.
Recall that the dimension $L$ is $2\sim5$ which is the number of effective paths in mmWave systems and 4 comes from the geometric parameters: AoD, time delay, path loss, and phase. 

\subsection{Transformer-based CSI Feedback and Reconstruction}
\subsubsection{Parametric Channel Quantization}
Once the parametric CSI $\widetilde{\mathbf{P}}_{t+\Delta t}$ is acquired at the Transformer-based encoder, UE quantizes and feeds back the acquired parametric CSI to the BS.
In the quantization process, we use the uniform quantization codebook:
\begin{align}
    \mathcal{C}_{\theta} &= \bigg\{\frac{2\pi q}{2^{Q_{\theta}}}\mid q=0,1,\cdots,2^{Q_{\theta}}-1\bigg\} \\
    \mathcal{C}_{\tau} &= \big\{\frac{\tau_{\text{max}}q}{2^{Q_{\tau}}}\mid q=0,1,\cdots,2^{Q_{\tau}}-1\big\}  \\
    \mathcal{C}_{\beta} &= \bigg\{\frac{\beta_{\text{max}}q}{2^{Q_{\beta}}}\mid q=0,1,\cdots,2^{Q_{\beta}}-1\bigg\} \\
    \mathcal{C}_{\phi} &= \bigg\{\frac{2\pi q}{2^{Q_{\phi}}}\mid q = 0, 1, \cdots, 2^{Q_{\phi}}-1\bigg\},
\end{align}
where $Q_{\theta}$, $Q_{\tau}$, $Q_{\beta}$, and $Q_{\phi}$ are the numbers of quantization bits for AoDs, time delays, path losses, and phases, respectively.
Let $\mathcal{C}_{\Bar{\boldsymbol{\theta}}}=\prod_{i=1}^L \mathcal{C}_{\theta}$ be the codebook of AoDs, then the chosen codeword is
\begin{align}
    \mathbf{c}_{\hat{q}_{\theta}} = \arg \min_{\mathbf{c}_{\theta} \in \mathcal{C}_{\Bar{\boldsymbol{\theta}}}} \Vert \boldsymbol{\theta}_{t+\Delta t}-\mathbf{c}_{\theta} \Vert^2.
\end{align}
Similarly, the UE also quantizes time delays $\boldsymbol{\tau}_{t+\Delta t}$, path losses $\boldsymbol{\beta}_{t+\Delta t}$, and phases $\boldsymbol{\phi}_{t+\Delta t}$ and then feeds back the indices of the chosen codeword $\mathbf{c}=\{\mathbf{c}_{\hat{q}_{\theta}},\mathbf{c}_{\hat{q}_{\tau}}, \mathbf{c}_{\hat{q}_{\beta}}, \mathbf{c}_{\hat{q}_{\phi}}\}$ to the BS.

\subsubsection{Transformer-based Channel Reconstruction at BS}
After receiving the quantized codeword $\mathbf{c}$, the decoder at the BS restores the AoDs $\boldsymbol{\hat{\theta}}=[\hat{\theta}_{1},\cdots,\hat\theta_{L}]$, the time delays $\boldsymbol{\hat{\tau}}=[\hat{\tau}_{1},\cdots,\hat\tau_{L}]$, the path losses $\boldsymbol{\hat{\beta}}=[\hat{\beta}_{1},\cdots,\hat\beta_{L}]$, and the phases $\boldsymbol{\hat{\phi}}=[\hat{\phi}_{1},\cdots,\hat\phi_{L}]$ from the codeword $\mathbf{c}$:
\begin{align}
    \mathbf{\hat{P}}_{t+\Delta t} &= \left\{\hat{\boldsymbol{\theta}}, \hat{\boldsymbol{\tau}}, \hat{\boldsymbol{\beta}}, \hat{\boldsymbol{\phi}}\right\} \\
    &= f_{\text{de-qnt}}(\mathbf{c}).
\end{align}
Recall that the mmWave MIMO channel $\hat{\mathbf{H}}_{t+\Delta t}$ is reconstructed via the Transformer-based decoder and parametric CSI $\hat{\mathbf{P}}_{t+\Delta t}$,
\begin{align}
    \mathbf{\hat{H}}_{t+\Delta t}
    =f_{\text{dec}}(\mathbf{\hat{P}}_{t+\Delta t}, \boldsymbol{\gamma}).
\end{align}
The decoder architecture is similar to the Transformer-based encoder in UE, except for the leaky-ReLU and 2D convolution layer at the end.
Specifically, $\hat{\mathbf{P}}_{t+\Delta t}$ initially passes through Transformer-based decoder to obtain the output vector $\mathbf{Z}_{\text{dec}}$.
By extracting the correlated features of the parametric CSI $\hat{\mathbf{P}}_{t+\Delta t}$, the decoder reconstructs the original channel matrix $\hat{\mathbf{H}}_{t+\Delta t}$.
Then, the output vector $\mathbf{Z}_{\text{dec}}$ passes through the leaky-ReLU layer to impose non-linearity to the transformed hidden units:
\begin{align}
   \check{\mathbf{Z}}_{\text{dec}} = f_{\text{leaky-ReLU}}(\mathbf{Z}_{\text{dec}})
\end{align}
where $f_{\text{leaky-ReLU}}(x) = \text{max}(0.1x, x)$ is the leaky-ReLU function. 
Finally, we use the 2D convolution layer to map the extracted channel features onto the channel matrix $\hat{\mathbf{H}}_{t+\Delta t}\in \mathbb{R}^{2\times N_f \times N_t}$, which is
\begin{align}
   \hat{\mathbf{H}}_{t+\Delta t} = f_{\text{2D-conv}}(\check{\mathbf{Z}}_{\text{dec}}, \mathbf{W}_{\text{2D-conv}})
\end{align} 
where $\mathbf{W}_{\text{2D-conv}}$ denotes the network parameters of the 2D convolution layer.

\section{Parameter-selective Feedback Bit Allocation}

Due to the quantization process, mismatch between the true channel parameters $\mathbf{P}$ and the quantized channel parameters $\hat{\mathbf{P}}$ is unavoidable.
One notable observation is that the impact of quantization on channel reconstruction performance varies for each parameter. 
To demonstrate this behavior, we measure the channel reconstruction error caused by the quantization of each channel parameter in Table~\ref{NMSE_fb}.
One can see that the channel reconstruction error caused by the angle quantization is significantly larger than that caused by the phase quantization. 
To minimize the channel quantization error in the practical limited feedback scenario, the feedback bits should be properly allocated among the channel parameters.

We now formalize the problem of minimizing the channel quantization error. 
Let $\Delta\boldsymbol{\theta}=\boldsymbol{\theta}-\hat{\boldsymbol{\theta}}$, $\Delta\boldsymbol{\tau}=\boldsymbol{\tau}-\hat{\boldsymbol{\tau}}$, $\Delta\boldsymbol{\beta}=\boldsymbol{\beta}-\hat{\boldsymbol{\beta}}$, and $\Delta\boldsymbol{\phi}=\boldsymbol{\phi}-\hat{\boldsymbol{\phi}}$ be the quantization distortion vectors of AoD, time delay, path loss, and phase, respectively. 
Since the channel parameters are quantized using the uniform quantization codebook, elements of quantization distortion vectors are also uniformly distributed. That is, 
\begin{align}
    &\Delta\theta_{l}\sim\text{Unif}\Big[-\frac{\pi}{2^{Q_{\theta}}},\frac{\pi} {2^{Q_{\theta}}}\Big]\label{d_theta} \\
    &\Delta\tau_{l}\sim\text{Unif}\Big[-\frac{\tau_{\mathrm{max}}}{2^{Q_{\tau}+1}},\frac{\tau_{\mathrm{max}}}{2^{Q_{\tau}+1}}\Big] \label{d_tau}\\
    &\Delta\beta_{l}\sim\text{Unif}\big[-\frac{\beta_{\mathrm{max}}}{2^{Q_{\beta}+1}},\frac{\beta_{\mathrm{max}}}{2^{Q_{\beta}+1}}\Big] \label{d_beta} \\
    &\Delta\phi_{l}\sim\text{Unif}\Big[-\frac{\pi}{2^{Q_{\phi}}}\frac{\pi}{2^{Q_{\phi}}}\Big].\label{d_phi} 
\end{align}
Also, let $\Delta\mathbf{H}=\mathbf{H}-\hat{\mathbf{H}}$ be the channel quantization distortion matrix. 
Then, the feedback bit allocation problem $\mathcal{P}$ to find out the optimal feedback bits $Q_{\theta}$, $Q_{\tau}$, $Q_{\beta}$, and $Q_{\phi}$ minimizing the channel quantization error $\mathbb{E}\big[\lVert\Delta\mathbf{H}\rVert^{2}\big]$ is expressed as
\begin{subequations}
    \begin{align}
    \mathcal{P}:\min_{Q_{\theta},Q_{\tau},Q_{\beta},Q_{\phi}}\,&\mathbb{E}\big[\lVert\Delta\mathbf{H}\rVert_{\mathrm{F}}^{2}\big]\\
        \text{s.t.}\quad\quad& Q_{\theta}+Q_{\tau}+Q_{\beta}+Q_{\phi}=Q
    \end{align}
\end{subequations}
where $Q$ is the total number of feedback bits. 
Note that $\mathcal{P}$ is a non-convex combinatorial optimization problem.
When the brute-force approach is applied to solve $\mathcal{P}$, it requires the evaluation of all possible $\big({Q\choose 4}\big)={Q+3\choose 4}$ combinations of $(Q_{\theta},Q_{\tau},Q_{\beta},Q_{\phi})$ where $\big({n\choose k}\big)$ denotes the $k$-combination with repetition from a set of size $n$.
For example, when $Q=20$, the total number of possible combinations is ${23\choose 4}=8855$.

To find a tractable solution of $\mathcal{P}$, we investigate the impact of quantization of each channel parameter on $\mathbb{E}\big[\lVert\Delta\mathbf{H}\rVert_{\mathrm{F}}^{2}\big]$.
Then we find out the optimal feedback bit allocation to minimize the channel quantization error.
In the following lemma, we approximate $\Delta \mathbf{H}$ as a function of $\Delta \boldsymbol{\theta}$, $\Delta \boldsymbol{\tau}$, $\Delta\boldsymbol{\beta}$, and $\Delta \boldsymbol{\phi}$.
\begin{lemma} \label{lem:1}
The channel quantization distortion $\Delta\mathbf{H}=\mathbf{H}-\hat{\mathbf{H}}=\big[\Delta\mathbf{h}[1]\cdots\Delta\mathbf{h}[N_{f}]\big]^{\textup{H}}$ between the true channel matrix $\mathbf{H}$ and the reconstructed channel matrix $\hat{\mathbf{H}}$ can be expressed as a function of the quantization distortions vectors $\Delta\boldsymbol{\theta}$, $\Delta\boldsymbol{\tau}$, $\Delta\boldsymbol{\beta}$, and $\Delta\boldsymbol{\phi}$:
\begin{align}
    \Delta\mathbf{h}[s]\approx&(\nabla_{\boldsymbol{\theta}}\mathbf{h}[s])\Delta\boldsymbol{\theta}+(\nabla_{\boldsymbol{\tau}}\mathbf{h}[s])\Delta\boldsymbol{\tau}+(\nabla_{\boldsymbol{\beta}}\mathbf{h}[s])\Delta\boldsymbol{\beta} \nonumber \\ &+(\nabla_{\boldsymbol{\phi}}\mathbf{h}[s])\Delta\boldsymbol{\phi}
\end{align}
where 
\begin{align}
    (\nabla_{\boldsymbol{\theta}}\mathbf{h}[s])\Delta\boldsymbol{\theta} =& \big(\mathbf{R}_{\theta}\odot\mathbf{A}_t(\boldsymbol{\theta})\big)\mathrm{diag}(e^{j(\boldsymbol{\phi}-2\pi f_{s}\boldsymbol{\tau})})\boldsymbol{\beta} \\
    (\nabla_{\boldsymbol{\tau}}\mathbf{h}[s])\Delta\boldsymbol{\tau} =& (\mathbf{R}_{\tau}[s]\odot\mathbf{A}_t(\boldsymbol{\theta}))\mathrm{diag}(e^{j(\boldsymbol{\phi}-2\pi f_{s}\boldsymbol{\tau})})\boldsymbol{\beta}\\
    (\nabla_{\boldsymbol{\beta}}\mathbf{h}[s])\Delta\boldsymbol{\beta} =& \mathbf{A}_t(\boldsymbol{\theta})\mathrm{diag}(e^{j(\boldsymbol{\phi}-2\pi f_{s}\boldsymbol{\tau})})\Delta\boldsymbol{\beta}\\
    (\nabla_{\boldsymbol{\phi}}\mathbf{h}[s])\Delta\boldsymbol{\phi} =& (\mathbf{R}_{\phi}\odot\mathbf{A}_t(\boldsymbol{\theta}))\mathrm{diag}(e^{j(\boldsymbol{\phi}-2\pi f_{s}\boldsymbol{\tau})})\mathrm{diag}(\boldsymbol{\beta}).
\end{align}
Also, $\mathbf{R}_{\theta}\in\mathbb{C}^{N_{t}\times L}$, $\mathbf{R}_{\tau}[s]\in\mathbb{C}^{N_{t}\times L}$, and $\mathbf{R}_{\phi}\in\mathbb{C}^{N_{t}\times L}$ are coefficient matrices given by
\begin{align}
    \mathbf{R}_{\theta}&=-j\frac{2\pi d}{\lambda}\big[\cos{\theta_{1}}\Delta\theta_{1}\mathbf{n}_{N_{t}},\cdots,\cos{\theta_{L}}\Delta\theta_{L}\mathbf{n}_{N_{t}}\big] \label{R_theta}\\
    \mathbf{R}_{\tau}[s]&=-j2\pi f_{s}[\Delta\tau_{1}\mathbf{1}_{N},\cdots,\Delta\tau_{L}\mathbf{1}_{N}]\\
    \mathbf{R}_{\phi}&=j[\Delta\phi_{1}\mathbf{1}_{N},\cdots,\Delta\phi_{L}\mathbf{1}_{N}].
\end{align}
where $\mathbf{n}_{N_{t}}=[0,1,\cdots,N_{t}-1]$.
\end{lemma}
\begin{proof}
    See Appendix~\hyperref[app:A]{A}.
\end{proof}

Using Lemma~\ref{lem:1} and the distributions of $\Delta\boldsymbol{\theta}$, $\Delta\boldsymbol{\tau}$, $\Delta\boldsymbol{\beta}$, and $\Delta\boldsymbol{\phi}$ in (\ref{d_theta})-(\ref{d_phi}), we can express the expected channel quantization error $\mathbb{E}[\lVert\Delta\mathbf{H}\rVert_{\mathrm{F}}^{2}]$ as a function of  $Q_{\theta}, Q_{\tau}, Q_{\beta}$, and $Q_{\phi}$. 
\begin{theorem}\label{thm:1}
The expected channel quantization error $\mathbb{E}\big[\lVert\Delta\mathbf{H}\rVert_{\mathrm{F}}^{2}\big]$ is expressed as the sum of angle-based channel quantization term $C_{\theta}$, time delay-based channel quantization term $C_{\tau}$, path loss-based channel quantization term $C_{\beta}$, and phase-based channel quantization term $C_{\phi}$ as
\begin{equation}
    \mathbb{E}\big[\lVert\Delta\mathbf{H}\rVert_{\textup{F}}^{2}\big] \approx C_{\theta}+C_{\tau}+C_{\beta}+C_{\phi}
\end{equation}
where
\begin{align}
    C_{\theta} &=\frac{\pi^{4}d^{2}LN_{f}N_{t}(N_{t}-1)\beta_{\textup{max}}^{2}}{36\cdot 2^{2Q_{\theta}}\lambda^{2}} \label{distort-angle}\\
    C_{\tau} &= \frac{\pi^{2}LN_{t}\tau_{\textup{max}}^{2}\beta_{\textup{max}}^{2}}{9\cdot 2^{2(Q_{\tau}+1)}}\sum_{s=1}^{N_{f}}f_{s}^{2} \label{distort-t} \\
    C_{\beta} &=\frac{LN_{f}N_{t}\beta_{\textup{max}}^{2}}{3\cdot 2^{2(Q_{\beta}+1)}} \label{distort-a}\\
    C_{\phi} &= \frac{\pi^{2}LN_{f}N_{t}\beta_{\textup{max}}^{2}}{36\cdot 2^{2Q_{\phi}}}.
    \label{distort-p}
\end{align}
\end{theorem}
\begin{proof}
    See Appendix~\hyperref[app:B]{B}.
\end{proof}
One can clearly see that the impacts of channel parameter quantization on the channel quantization error $\mathbb{E}\big[\lVert\Delta\mathbf{H}\rVert_{\textup{F}}^{2}\big]$ are different.
Thus, to minimize the channel quantization error,  we need to properly allocate feedback bits such that larger number of bits is allocated to the parameter having greater impact on the performance of channel reconstruction.

\begin{table}[t]
\centering
\caption{Channel reconstruction error caused by the quantization of different channel parameters.}
\begin{tabular}{|C{2cm}|C{1.5cm}|C{1.5cm}|C{1.5cm}|}
\hline
\textbf{} & \textbf{6 bits} & \textbf{7 bits} & \textbf{8 bits} \\ \hline
Path loss $\beta$     & -30.77        & -36.92         & -42.84        \\ \hline
Phase $\phi$    & -41.81        & -47.99        & -54.04        \\ \hline
Time delay $\tau$    & -6.76         & -12.79        & -19.49        \\ \hline
AoD $\theta$    & -5.09         & -10.50        & -16.65        \\ \hline
\end{tabular}
\label{NMSE_fb}
\end{table}

The feedback bit allocation problem $\mathcal{P}$ can be reformulated as
\begin{subequations}
    \begin{align}               \mathcal{P'}:\min_{Q_{\theta},Q_{\tau},Q_{\beta},Q_{\phi}}\,&C_{\theta}+C_{\tau}+C_{\beta}+C_{\phi} \label{ob-func}\\
        \text{s.t.}\quad\quad & Q_{\theta}+Q_{\tau}+Q_{\beta}+Q_{\phi}=Q.
    \end{align}
\end{subequations}
Using the Arithmetic-Geometric Mean Inequality\footnote{
$C_{\theta}+C_{\tau}+C_{\beta}+C_{\phi}$ satisfies Arithmetic-Geometric Mean Inequality when $x=2^{-2Q_{\theta}}$, $y=2^{-2Q_{\tau}}$, $z=2^{-2Q_{\beta}}$, and $u=2^{-2Q_{\phi}}$, which is $\frac{ax+by+cz+du}{4} \geq \sqrt[4]{ax\cdot by\cdot cz\cdot du}$ where the equality holds in $ax=by=cz=du$.
}, quantization bits of the channel parameters $Q_{\theta}, Q_{\tau}, Q_{\beta}$, and $Q_{\phi}$ minimizing the loss function in~(\ref{ob-func}) is 
\begin{align}
    Q_{\theta} &= \frac{Q}{4}+\frac{1}{8}\log_2\Big(\frac{\pi^8 d^6 (N_t-1)^3}{3 \lambda^6 \tau_{\mathrm{max}}^2 \sum_{s=1}^{N_{f}}f^2_{s}}\Big) \label{Q_theta} \\
    Q_{\tau} &= \frac{Q}{4}+\frac{1}{8}\log_2\Big(\frac{\lambda^2 \tau^6_{\mathrm{max}} (\sum_{s=1}^{N_{f}}f^2_{s})^3}{3d^2 N_f^3(N_t-1)}\Big) \label{Q_tau} \\
    Q_{\beta} &= \frac{Q}{4}+\frac{1}{8}\log_2\Big(\frac{27\lambda^2 N_f}{\pi^8 d^2 \tau_{\mathrm{max}}^2 (\sum_{s=1}^{N_{f}}f_{s}) (N_t-1)}\Big) \label{Q_beta} \\
    Q_{\phi} &= \frac{Q}{4}+\frac{1}{8}\log_2\Big( \frac{\lambda^2 N_f}{3d^2 \tau_{\mathrm{max}}^2 (\sum_{s=1}^{N_{f}}f^2_{s})(N_t-1)}\Big). \label{Q_phi}
\end{align}

One might argue that the proposed COMPaCT requires an additional feedback overhead of transmitting the information of quantization bits (i.e., $Q_\theta$, $Q_\tau$, $Q_\beta$, and $Q_\phi$).
However, as shown from (\ref{Q_theta})-(\ref{Q_phi}), these quantization bits are not affected by the geometric channel parameters. 
This means that reporting of the parameter-selective information of quantization bits to the UE can be done occasionally (e.g., reporting via PDCCH~\cite{3gpp.38.211} during the initial access).

\section{Experimental Results}

\begin{table*}[t]
\begin{center}
\caption{System parameters and training settings.}
\captionsetup{skip=5pt, position = bottom}
\begin{tabular}{|m{5.5cm}|m{1.5cm}||m{5.5cm}|m{1.5cm}|}
\hline
\textbf{System Parameters} & \textbf{Value} & \textbf{Training Settings} & \textbf{Value} \\ \hline \hline
Number of BS Antenna $N_t$       & 64             & Train data                                     & 80,000         \\ \hline
Number of UE Antenna $N_r$       & 1              & Validation data                                & 20,000         \\ \hline
Carrier frequency $f_c$          & 28 GHz         & Test data                                      & 20,000         \\ \hline
System bandwidth $B$             & 100 MHz        & Batch size                                     & 16             \\ \hline
Number of OFDM subcarriers $N_c$ & 1024           & Embedding dimension in attention module        & 512            \\ \hline
Maximum number of path $L_{\mathrm{max}}$       & 10             & Number of attention head                       & 8              \\ \hline
Mobility of UE $v$               & 3 km/h& Initial learning rate                          & 0.002          \\ \hline
Feedback period $T_s$      & 10 ms          & Learning rate decay period                     & 50 epochs      \\ \hline
Observed number of time slots $W$ & 20             & Learning rate decay coefficient                & 0.1            \\ \hline
\end{tabular}
\label{table:2}
\end{center}
\end{table*}

\subsection{Simulation Setup}
In our simulations, we consider the mmWave massive MIMO systems where the BS and UE are equipped with $N_t=64$ antennas and a single antenna, respectively.
We generate the channel using the ray-tracing method described in 3GPP TR 38.901~\cite{3gpp.38.901}.
Specifically, the UE is located initially in a random position such that the distance from the BS is between $10\,\text{m}$ and $500\,\text{m}$.
Then UE moves in a straight line at a speed of $v = 3$ km/h, where the direction of the trajectory is sampled in a uniform distribution [0, 2$\pi$).
The height of BS and UE is $25\,\text{m}$ and $1.5\,\text{m}$, respectively.
We set the carrier frequency to $f_c = 28$ GHz, the total number of OFDM subcarriers to ${N_c} = 1024$, and the total bandwidth to $B = 100$ MHz.
Considering the severe path loss and strong directivity of the mmWave channel, we limit the maximum number of paths to $L_{\mathrm{max}}=10$.
The training, validation, and testing sets consist of $80,000$, $20,000$, and $20,000$, respectively.
The detailed system parameters and training parameters are summarized in Table~\ref{table:2}~\cite{samimi20163}.

\begin{algorithm}[t]
\caption{Training process of COMPaCT}\label{alg:cap}
\begin{algorithmic}

    \Require Vectorized $w$-step channel sequence $\widetilde{\mathbf{H}}_{\text{input}}$    
\For {$\text{iter}=1$ to $N_{\text{iteration}}$}
    \State Obtain the input matrix $\widetilde{\mathbf{H}}_{\text{input}}$ of the encoder by \eqref{vectorize}
    \State Predict the parametric channel $\widetilde{\mathbf{P}}_{t+\Delta t} \gets f_{\text{enc}}(\widetilde{\mathbf{H}}_{\text{input}}, \boldsymbol{\delta})$
    \State Quantize the codeword $\mathbf{c} \gets f_{\text{qnt}}(\widetilde{\mathbf{P}}_{t+\Delta t})$
    \State De-quantize the parametric channel $\mathbf{\hat{P}}_{t+\Delta t} \gets \{\hat{\boldsymbol{\theta}}, \hat{\boldsymbol{\tau}}, \hat{\boldsymbol{\beta}}, \hat{\boldsymbol{\phi}}\} \gets f_{\text{de-qnt}}(\mathbf{c})$
    \State Recover the channel matrix $\mathbf{\hat{H}}_{t+\Delta t}
        \gets f_{\text{dec}}(\mathbf{\hat{P}}_{t+\Delta t}, \boldsymbol{\gamma})$

    \State Update the model parameters $\boldsymbol{\delta}, \boldsymbol{\gamma}$
    \State $\boldsymbol{\delta} \gets \boldsymbol{\delta} - r \nabla_{\boldsymbol{\delta}} L_{\boldsymbol{\delta}}(\mathbf{\hat{H}}_{t+\Delta t}, \mathbf{H}_{t+\Delta t})$
    \State $\boldsymbol{\gamma} \gets \boldsymbol{\gamma} - r \nabla_{\gamma} L_{\boldsymbol{\gamma}}(\mathbf{\hat{H}}_{t+\Delta t}, \mathbf{H}_{t+\Delta t})$
\EndFor
    \Ensure Predicted channel matrix in time $t+\Delta t$ $\mathbf{\hat{H}}_{t+\Delta t}$
\end{algorithmic}
\end{algorithm}

An integral part of the proposed COMPaCT is the training process optimizing the network parameters $\boldsymbol{\delta}=\left\{\mathbf{W}^Q, \mathbf{W}^K, \mathbf{W}^V, \mathbf{W}_{c}  \right\}$ at the encoder and $\boldsymbol{\gamma} = \left\{ \mathbf{W}_{\text{2D-conv}} \right\} $ at the decoder.
To train the Transformer-based encoder, we first acquire the intermediate channel $\tilde{\mathbf{H}}$ from the obtained channel parameter $\widetilde{\mathbf{P}}$ in~(\ref{parameterize}).
The loss function of the Transformer-based encoder is the NMSE between the intermediate channel $\tilde{\mathbf{H}}$ and the true channel $\mathbf{H}$, given by 
\begin{align}
    L(\boldsymbol{\delta}) = \frac{1}{N_B}\sum_{i=1}^{N_B} \frac{\Vert \tilde{\mathbf{H}}_{i} -\mathbf{H}_{i}\Vert_{\mathrm{F}}^2}{\Vert \mathbf{H}_{i}\Vert_{\mathrm{F}}^2}
    \label{loss_NMSE}
\end{align}
where $N_B$ is the batch size of the training dataset.
Note that the genie channel $\mathbf{H}$ is obtained by the geometric MIMO channel model in (\ref{channel}).
The loss function of the Transformer-based decoder is NMSE between the reconstructed channel at the decoder $\hat{\mathbf{H}}$ and the true channel $\mathbf{H}$, expressed as
\begin{align}
    L(\boldsymbol{\gamma}) = \frac{1}{N_B}\sum_{i=1}^{N_B} \frac{\Vert \hat{\mathbf{H}}_{i} -\mathbf{H}_{i}\Vert_{\mathrm{F}}^2}{\Vert \mathbf{H}_{i}\Vert_{\mathrm{F}}^2}.
\end{align}
Using these loss functions, we update the network parameters to find the optimal network parameters $\boldsymbol{\delta^{*}}$ and $\boldsymbol{\gamma^{*}}$ characterizing the mapping function $f_{\text{enc}}$ and $f_{\text{dec}}$, respectively.
The overall training process of COMPaCT is summarized in Algorithm~\ref{alg:cap}.

We compare the channel feedback performance of the proposed COMPaCT with eight baseline channel feedback techniques: 
\begin{enumerate}
    \item $\textit{LASSO}$~\cite{daubechies2004iterative}, a well-known compressive sensing (CS)-based approach that compresses the channel matrix into a sparse matrix using an L1-regularizer.
    \item $\textit{CsiNet}$~\cite{wen2018deep}, which uses a CNN-based autoencoder to compress and restore the channel matrix.
    \item $\textit{TransNet}$~\cite{cui2022transnet}, which exploits Transformer in the autoencoder to compress and recover the channel matrix.
    \item $\textit{DLRec}$~\cite{han2019tracking}, where the UE estimates and feeds back downlink path losses and phases via least squares (LS). Note that the path angles and time delays are acquired at BS using the angular-delay reciprocity.
    \item $\textit{CVNN}$~\cite{wu2021channel}, which extracts the significant elements in the angle-delay domain channel response matrix (ADCRM) to predict the ADCRM of the downlink channel.
    \item $\textit{Pruning}$~\cite{mashhadi2021pruning}, which estimates the channel using a specially designed pilot signal and non-local attention module.
    \item $\textit{Parallel Transformer-based CSI network (PTNet)}$~\cite{jiang2022accurate}, which processes the previous channel data in parallel using Transformer to predict future channels.
    \item $\textit{CsiNet-LSTM}$~\cite{wang2018deep}, which utilizes the LSTM architecture to predict the future channels.
\end{enumerate}
In the quantization process, we use the uniform quantizer exclusively.

\subsection{Simulation Results}

\begin{figure}[t]
    \centering 
    \includegraphics[width=1\linewidth]{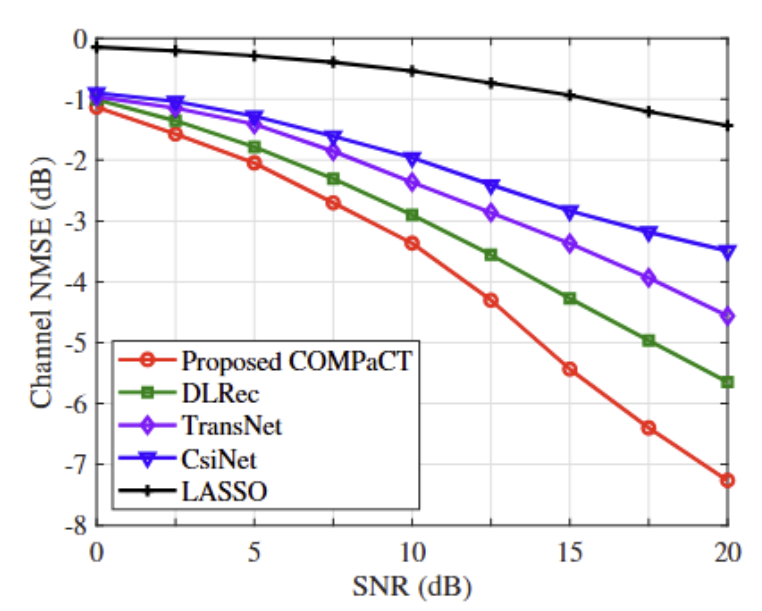}
    \caption{Channel NMSE without quantization as a function of SNR (dB), $v=3\,$km/h.}
    \label{result1}
\end{figure}

In Fig.~\ref{result1}, we evaluate the performance of the CSI feedback without quantization as a function of SNR.
As a performance metric, we use the channel NMSE defined in \eqref{loss_NMSE}.
We observe that the proposed COMPaCT outperforms the conventional channel feedback techniques by a large margin, especially at high SNR regime.
For example, when SNR $=20\,\text{dB}$, COMPaCT achieves a reduction of more than 3.5$\,$dB in NMSE over CsiNet.
This is because Transformer can capture both short- and long-term temporal correlations using the attention mechanism while CNN struggles with extracting long-term correlations.
Even when compared to DLRec, the NMSE gain of COMPaCT is more than 1.5$\,$dB. 
This is because COMPaCT can fastly estimate changing complex gain (i.e., $\alpha_{l}=\beta_{l} e^{j \phi_{l}}$) via Transformer-based encoder.

\begin{figure}[t]
    \centering 
    \includegraphics[width=1\linewidth]{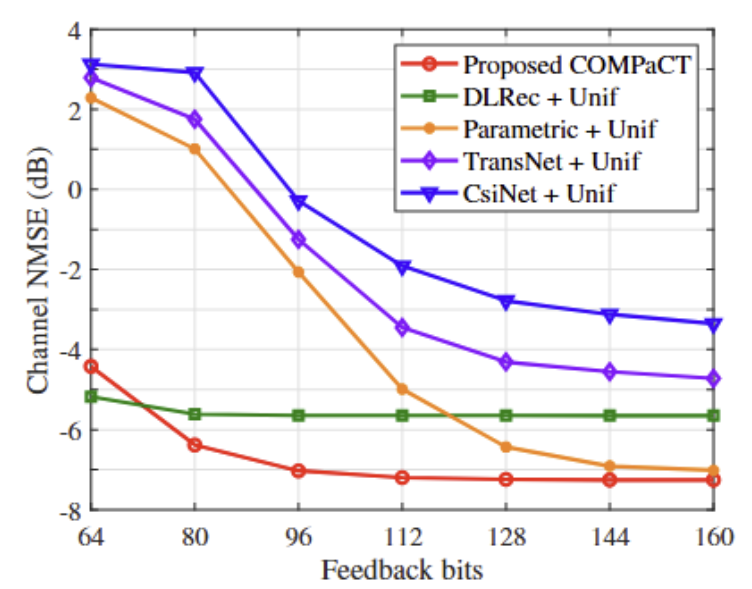}
    \caption{Channel NMSE with quantization as a function of feedback bits, $v=3\,$km/h, SNR = $20\,$dB.}
    \label{result3}
\end{figure}
In Fig.~\ref{result3}, we evaluate the CSI feedback performance with quantization in NMSE as a function of total feedback bits $Q$. 
In the conventional schemes, we use the uniform bit allocation.
We observe that COMPaCT outperforms conventional schemes in all feedback bit regimes by a significant margin, in particular the NMSE gain of COMPaCT is phenomenal for small $Q$.
For example, when the total number of feedback bits $Q$ is 96, COMPaCT exhibits 7$\,$dB and 5.5$\,$dB lower in NMSE compared to CsiNet and TransNet, respectively.
In order to further investigate the efficacy of the parameter selective bit allocation method, we compare the proposed COMPaCT with the parametric CSI feedback technique based on uniform bit allocation. We observe that COMPaCT achieves more than 1.6$\,$dB NMSE gain on average over the parametric CSI feedback with uniform bit allocation. This is because COMPaCT minimizes the distortion of the channel NMSE by the parameter selective bit allocation.

\begin{figure}[t]
    \centering 
    \includegraphics[width=1\linewidth]{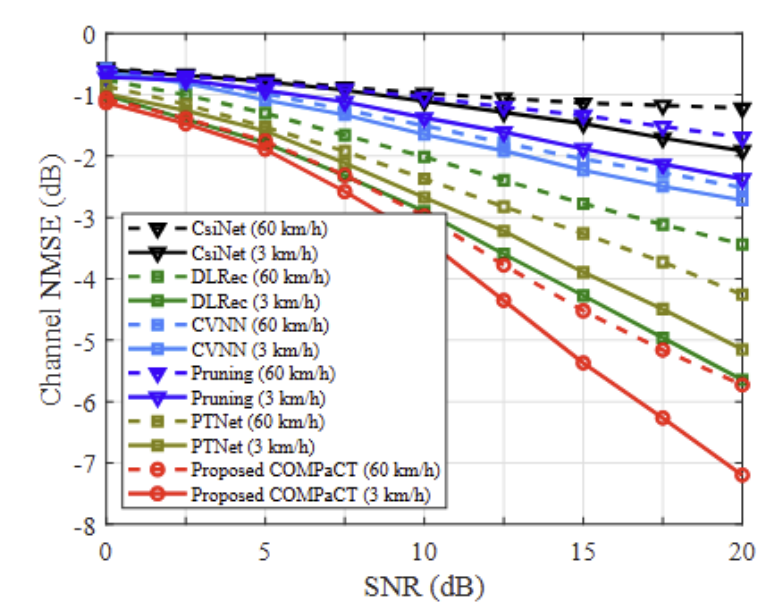}
    \caption{Channel NMSE with quantization as a function of SNR, in various UE mobility scenarios, $Q$ = 112.}
    \label{result4}
\end{figure}

In Fig.~\ref{result4}, we evaluate the NMSE performance of COMPaCT in two UE mobility scenarios ($v = \{3, 60\}$$\,$km/h).
We observe that COMPaCT shows the best (lowest) NMSE in various mobility scenarios.
For example, COMPaCT achieves a 3$\,$dB reduction in channel NMSE compared to other conventional methods since COMPaCT can capture the variation of the channel parameters from the previous channel sequence.
We also observe that COMPaCT is robust to UE mobility. We note that the angular-delay reciprocity does not hold for non-stationary scenario (i.e., $v=\,$60$\,$km/h), so that DLRec suffers from NMSE degradation caused by the channel aging effect. 
On the other hand, COMPaCT estimates the geometric channel parameters at the desired time instant by extracting the correlated features from both temporally adjacent and distant channels. 
Indeed, compared to the stationary scenario (i.e., $v=$ 3 km/h), COMPaCT achieves a higher NMSE gain over DLRec in non-stationary scenario. 

\begin{figure}[t]
    \centering 
    \includegraphics[width=1\linewidth]{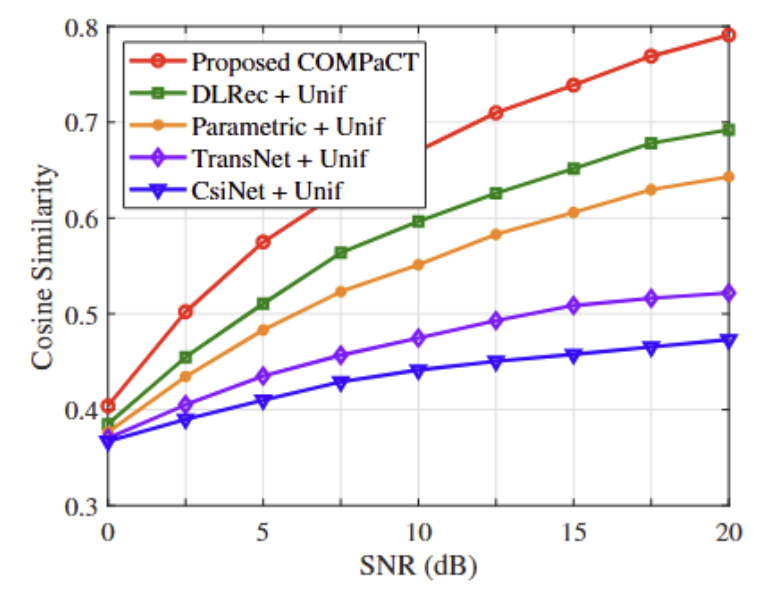}
    \caption{Cosine similarity with quantization as a function of SNR (dB), $v=60\,$km/h, $Q$ = 112.}
    \label{result_cos}
\end{figure}

In Fig.~\ref{result_cos}, we plot the cosine similarity, a metric to measure the quality of the beamforming vector, as a function of SNR.
Cosine similarity can be expressed as 
\begin{align}
    \rho = \mathbb{E}\left[\frac{1}{N_f} \sum_{s=1}^{N_f} \frac{\vert\hat{\mathbf{h}}^H[s] \mathbf{h}[s]\vert}{\Vert\hat{\mathbf{h}}[s]\Vert_2 \Vert \mathbf{h}[s]\Vert_2}\right].
\end{align}
One can observe that COMPaCT achieves the best (largest) correlation between the reconstructed channel $\hat{\mathbf{h}}[s]$ and the true channel $\mathbf{h}[s]$.
Since the beamforming gain is maximized when the beams are properly aligned with the signal propagation paths, one can deduce that beamforming gain of COMPaCT will be larger than the conventional techniques.

\begin{figure}[t]
    \centering 
    \includegraphics[width=1\linewidth]{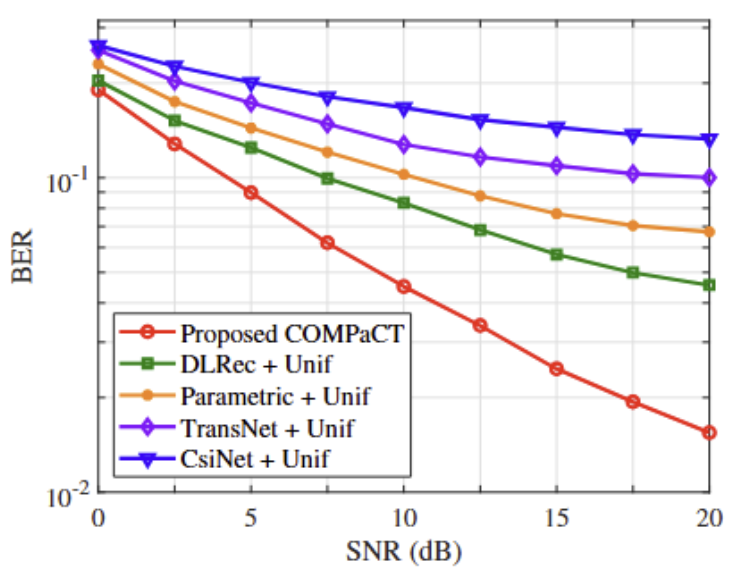}
    \caption{BER performance as a function of SNR, $v = 3\,$km/h, $Q$ = 112.}
    \label{result6}
\end{figure}
In Fig.~\ref{result6}, we compare the bit error rate (BER) of COMPaCT with conventional schemes.
We adopt the beamforming vector $\mathbf{r}_{t+\Delta t}[s] = \mathbf{\hat{h}}_{t+\Delta t}[s]/||\mathbf{\hat{h}}_{t+\Delta t}[s]||_2$ \cite{dutta2019case} and use the QPSK symbol with hard decoding.
We observe that COMPaCT achieves significant gain in BER performance in all SNR regimes.
For example, COMPaCT requires SNR $=10\,$dB to achieve BER $=0.05$, but DLRec achieves the same performance at SNR $=20\,$dB.

In Fig.~\ref{timeslots}, we compare the channel NMSE of COMPaCT as a function of the number of observed timeslots.
We observe that the proposed COMPaCT achieves more than $1\,$dB gain over the LSTM-based method.
This gap becomes even more significant when the number of timeslots is 15, showing more than $2\,$dB gain over the conventional method.
This is because the proposed COMPaCT employs a Transformer architecture that reflects relationships among all elements within a sequence simultaneously, and leverages channel parameters with longer coherence times to estimate the downlink channel.

\begin{figure}[t]
    \centering 
    \includegraphics[width=1\linewidth]{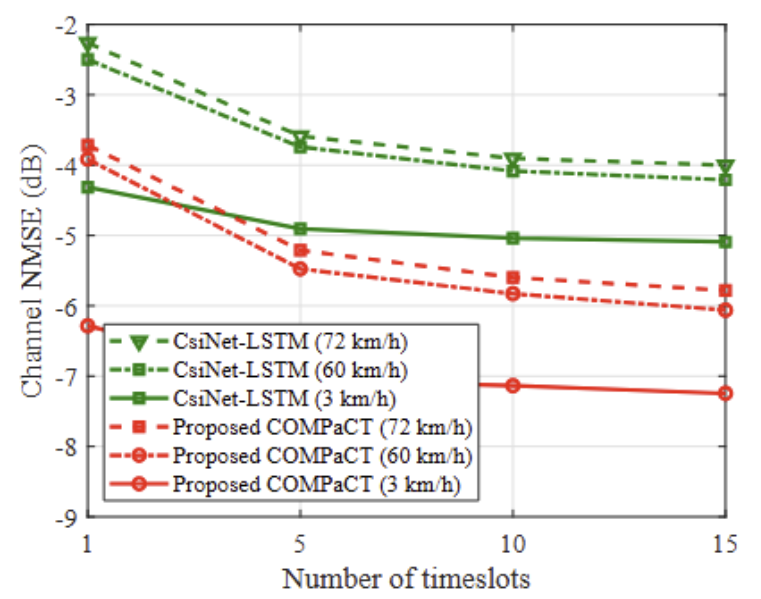}
    \caption{Channel NMSE with quantization as a function on timeslots, SNR = 20\,dB, $Q$ = 112.}
    \label{timeslots}
\end{figure}

To demonstrate the robustness of COMPaCT in various wireless environments, in Fig.~\ref{ruralurban}, we evaluate the performance of COMPaCT under different path loss models (i.e., urban area and rural area) in 3GPP Release 14~\cite{3gpp.38.901}.
We use channel data from urban areas where LoS is frequently obstructed by structures such as buildings, and from rural areas where LoS is dominant.
Training, validation, and test datasets of each scenario consist of 40,000, 10,000, and 10,000 channel data, respectively.
Interestingly, we observe that COMPaCT works well even when the tested scenario is different from the training scenario. 
The main reason for this is that the channel feedback via COMPaCT relies heavily on UE mobility in mmWave bands rather than on the specific channel model.

\begin{figure}[t]
    \centering 
    \includegraphics[width=1\linewidth]{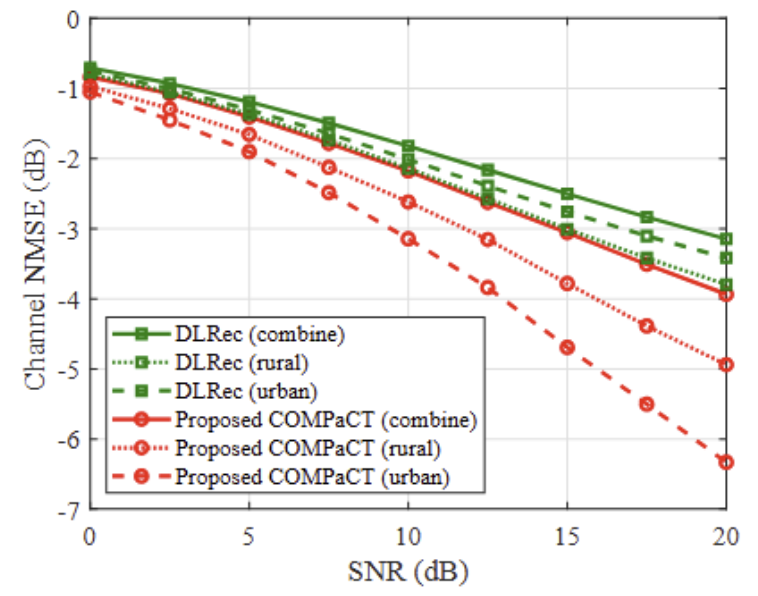}
    \caption{Channel NMSE in various channel environments.}
    \label{ruralurban}
\end{figure}

In Fig.~\ref{veloadd}, we evaluate the channel NMSE of COMPaCT in non-stationary scenarios where the UE speeds are $v=3\,$km/h (pedestrian), $v=60\,$km/h (vehicle), and $v=108\,$km/h (vehicle on highway).
We observe that the proposed COMPaCT achieves more than 1\,dB gain in terms of channel NMSE compared to DLRec in all mobility scenarios.
This result shows that the proposed COMPaCT supports accurate CSI feedback while UE is moving at extremely high speed, whether on the highway or aboard a train.

\begin{figure}[t]
    \centering 
    \includegraphics[width=1\linewidth]{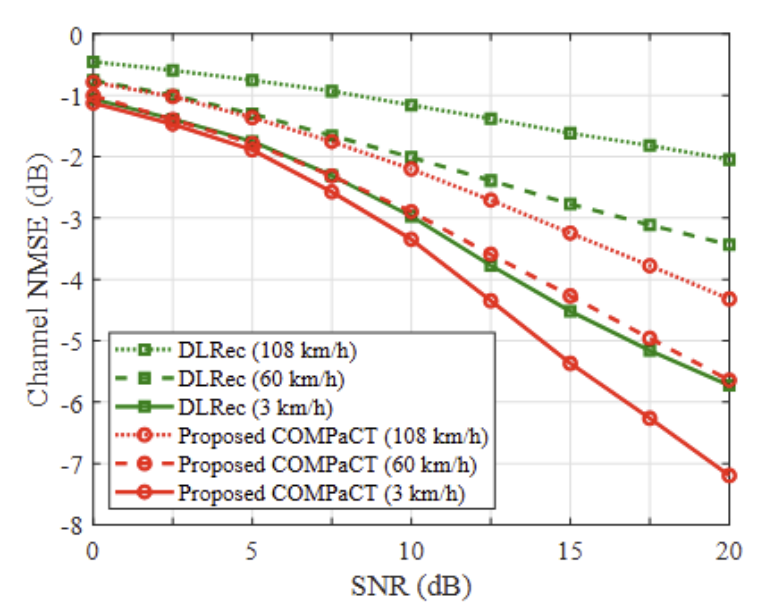}
    \caption{Channel NMSE with quantization in non-stationary scenarios.}
    \label{veloadd}
\end{figure}

\section{Conclusion}
In this paper, we proposed a DL-based parametric CSI feedback technique for mmWave massive MIMO systems.
Using the property that the mmWave MIMO channel can be expressed with a few channel parameters (e.g., angles, delays, and path gains), the proposed COMPaCT extracts geometric channel parameters from the MIMO channel matrix and then quantizes and feeds back the channel parameters instead of the full-dimensional MIMO channel. 
Then by combining the geometric channel parameters, UE reconstructs the MIMO channel matrix.
Intriguing feature of COMPaCT is that we use Transformer, the state-of-the-art DL model, for the channel parameter extraction and the MIMO channel reconstruction. 
By capturing the spatio-temporal correlation among the MIMO channel sequence via Transformer, COMPaCT can support accurate CSI feedback while assessing low feedback overhead.  
From the numerical results, we demonstrated that COMPaCT achieves a significant performance gain over the conventional CSI feedback schemes in terms of NMSE and BER.
We believe that the proposed COMPaCT can be an effective means for the channel acquisition in various upcoming 6G applications such as vehicle-to-everything (V2X) communications, non-terrestrial networks (NTN), and maritime communications.

\section*{Appendix A\\Proof of Lemma~\ref{lem:1}} \label{app:A}
First, the perturbation of $\mathbf{h}[s]$ with respect to $\boldsymbol{\theta}$ is
\begin{align}
    (\nabla_{\boldsymbol{\theta}}\mathbf{h}[s])\Delta\boldsymbol{\theta}
    =&\bigg[\frac{\partial \mathbf{h}[s]}{\partial \theta_{1}},\cdots,\frac{\partial \mathbf{h}[s]}{\partial \theta_{L}}\bigg][\Delta\theta_{1},\cdots,\Delta\theta_{L}]^{\mathrm{T}} \\
    =&\sum_{l=1}^{L}\frac{\partial \mathbf{h}[s]}{\partial \theta_{l}}\Delta\theta_{l}, \label{channel_perturb}
\end{align}
where the partial derivative of $\mathbf{h}[s]$ with respective to $\theta_{l}$ is given by
\begin{align}
    \frac{\partial \mathbf{h}[s]}{\partial \theta_{l}}
    =&\frac{\partial \mathbf{a}_{t}(\boldsymbol{\theta})}{\partial \theta_{l}}\mathrm{diag}(e^{j\boldsymbol{\phi}-2\pi f_{s}\boldsymbol{\tau}})\boldsymbol{\beta} \\
    =&\bigg[\mathbf{0}_{N},\cdots,\frac{\partial \mathbf{a}_{t}(\theta_{l})}{\partial \theta_{l}},\cdots,\mathbf{0}_{N_{t}}\bigg]\mathrm{diag}(e^{j\boldsymbol{\phi}-2\pi f_{s}\boldsymbol{\tau}})\boldsymbol{\beta} \\
    =&\bigg(\mathbf{e}_{l}^{\mathrm{T}}\otimes\frac{\partial \mathbf{a}_{t}(\theta_{l})}{\partial \theta_{l}}\bigg)\mathrm{diag}(e^{j\boldsymbol{\phi}-2\pi f_{s}\boldsymbol{\tau}})\boldsymbol{\beta}, \label{channel_gradient}
\end{align}
where $\mathbf{e}_{l}\in\mathbb{R}^{L}$ is the $L\times 1$ vector whose $l$-th element is one and all other elements are zero.
Note that the gradient of the array steering vector $\mathbf{a}_{t}(\theta_{l})$ with respect to $\theta_{l}$ is given by
\begin{align}
    \!\!\!\!\!\frac{\partial \mathbf{a}_{t}(\theta_{l})}{\partial \theta_{l}}&= \frac{\partial}{\partial \theta_{l}}\Big[1,e^{-j\frac{2\pi d\sin{\theta_{l}}}{\lambda}},\cdots,e^{-j(N_{t}-1)\frac{2\pi d\sin{\theta_{l}}}{\lambda}}\Big]^{\mathrm{T}} \\
    &=-j\frac{2\pi d\cos{\theta_{l}}}{\lambda}\Big[0,e^{-j\frac{2\pi d}{\lambda}\sin{\theta_{l}}},\cdots,\notag\\
    &\quad (N_{t}-1)e^{-j(N_{t}-1)\frac{2\pi d}{\lambda}\sin{\theta_{l}}}\Big]^{\mathrm{T}}  \\
    &=-j\frac{2\pi d\cos{\theta_{l}}}{\lambda}\bigg([0,\cdots,N_{t}-1]^{\mathrm{T}}\notag\\
    &\quad \odot\Big[1,e^{-j\frac{2\pi d}{\lambda}\sin{\theta_{l}}},\cdots,e^{-j(N_{t}-1)\frac{2\pi d}{\lambda}\sin{\theta_{l}}}\Big]^{\mathrm{T}}\bigg)  \\
    &=-j\frac{2\pi d\cos{\theta_{l}}}{\lambda}\big(\mathbf{n}_{N_{t}}\odot\mathbf{a}_{t}(\theta_{l})\big), \label{array_gradient}
\end{align}
where $\mathbf{n}_{N_{t}}=[0,1,\cdots,N_{t}-1]^{\mathrm{T}}\in\mathbb{R}^{N_{t}}$. 
By plugging \eqref{array_gradient} to \eqref{channel_gradient}, we have
\begin{align}
    \frac{\partial \mathbf{h}[s]}{\partial \theta_{l}}
    =&-j\frac{2\pi d\cos{\theta_{l}}}{\lambda}\Big(\mathbf{e}_{l}^{\mathrm{T}}\otimes\big(\mathbf{n}_{N_{t}}\odot\mathbf{a}_{t}(\theta_{l})\big)\Big)\notag\\
    &\times \mathrm{diag}(e^{j\boldsymbol{\phi}-2\pi f_{s}\boldsymbol{\tau}})\boldsymbol{\beta}.
\end{align}
Then, $\nabla_{\boldsymbol{\theta}}\mathbf{h}[s]\Delta\boldsymbol{\theta}$ in \eqref{channel_perturb} can be re-expressed as
\begin{align}
    \!\!\!\!\!\nabla_{\boldsymbol{\theta}}\mathbf{h}[s]\Delta\boldsymbol{\theta}  \!&=\sum_{l=1}^{L}\frac{\partial \mathbf{h}[s]}{\partial \theta_{l}}\Delta\theta_{l}  \\
    &= -j\frac{2\pi d}{\lambda}\sum_{l=1}^{L}\cos{\theta_{l}}\Big(\mathbf{e}_{l}^{\mathrm{T}}\otimes\big(\mathbf{n}_{N_{t}}\odot\mathbf{a}_{t}(\theta_{l})\big)\Big)\notag\\
    &\quad \times\mathrm{diag}(e^{j\boldsymbol{\phi}-2\pi f_{s}\boldsymbol{\tau}})\boldsymbol{\beta}\Delta\theta_{l}  \\
    &=-j\frac{2\pi d}{\lambda}\bigg(\sum_{l=1}^{L}\cos{\theta_{l}}\Big(\mathbf{e}_{l}^{\mathrm{T}}\otimes\big(\mathbf{n}_{N_{t}}\odot\mathbf{a}_{t}(\theta_{l})\big)\Big)\Delta\theta_{l}\bigg)\notag\\
    &\quad \times \mathrm{diag}(e^{j\boldsymbol{\phi}-2\pi f_{s}\boldsymbol{\tau}})\boldsymbol{\beta} \\
    &=-j\frac{2\pi d}{\lambda}\Big(\big[\cos{\theta_{1}}\Delta\theta_{1}\mathbf{n}_{N_{t}},\cdots,\cos{\theta_{L}}\Delta\theta_{L}\mathbf{n}_{N_{t}}\big]\notag\\
    &\quad\odot\big[\mathbf{a}_{t}(\theta_{1}),\cdots,\mathbf{a}_{t}(\theta_{L})\big]\Big)\mathrm{diag}(e^{-j2\pi f_{s}\boldsymbol{\tau}})\boldsymbol{\beta} \\
    &=\big(\mathbf{R}_{\theta}\odot\mathbf{A}_{t}(\boldsymbol{\theta})\big)\mathrm{diag}(e^{-j2\pi f_{s}\boldsymbol{\tau}})\boldsymbol{\beta},
\end{align}
where 
\begin{equation}
\mathbf{R}_{\theta}=-j\frac{2\pi d}{\lambda}\big[\cos{\theta_{1}}\Delta\theta_{1}\mathbf{n}_{N_{t}},\cdots,\cos{\theta_{L}}\Delta\theta_{L}\mathbf{n}_{N_{t}}\big].
\end{equation}
The perturbations of $\mathbf{h}[s]$ with respect to $\boldsymbol{\tau}$, $\boldsymbol{\beta}$, and $\boldsymbol{\phi}$ in Lemma~\ref{lem:1} can be obtained similarly.

\section*{Appendix B\\Proof of Theorem~\ref{thm:1}} \label{app:B}
Since $\Delta\boldsymbol{\theta}$, $\Delta\boldsymbol{\tau}$, $\Delta\boldsymbol{\beta}$, and $\Delta\boldsymbol{\phi}$ are independent, $\mathbb{E}\big[\lVert\Delta\mathbf{H}\rVert_{\text{F}}^{2}\big]$ can be expressed as
\begin{align}
\mathbb{E}\big[\lVert\Delta\mathbf{H}\rVert_{\text{F}}^{2}\big]\!=& \sum_{s=1}^{N_{f}}\mathbb{E}\big[\lVert\Delta\mathbf{h}[s]\rVert_{2}^{2}\big]  \\
=&\sum_{s=1}^{N_{f}}\Big(\mathbb{E}\big[\lVert(\nabla_{\boldsymbol{\theta}}\mathbf{h}[s])\Delta\boldsymbol{\theta}\rVert_{2}^{2}\big]+\mathbb{E}\big[\lVert(\nabla_{\boldsymbol{\tau}}\mathbf{h}[s])\Delta\boldsymbol{\tau}\rVert_{2}^{2}\big]\notag \\&+\mathbb{E}\big[\lVert(\nabla_{\boldsymbol{\beta}}\mathbf{h}[s])\Delta\boldsymbol{\beta}\rVert_{2}^{2}\big]
+\mathbb{E}\big[\lVert(\nabla_{\boldsymbol{\phi}}\mathbf{h}[s])\Delta\boldsymbol{\phi}\rVert_{2}^{2}\big]\Big).
\end{align}
First, the angle quantization term $\mathbb{E}\big[\lVert(\nabla_{\boldsymbol{\theta}}\mathbf{h}[s])\Delta\boldsymbol{\theta}\rVert_{2}^{2}\big]$ is computed as
\begin{align}
\mathbb{E}&\big[\lVert(\nabla_{\boldsymbol{\theta}}\mathbf{h}[s])\Delta\boldsymbol{\theta}\rVert_{2}^{2}\big] \nonumber \\ =&\mathbb{E}\Big[\big\lVert\big(\mathbf{R}_{\theta}\odot\mathbf{A}_{t}(\boldsymbol{\theta})\big)\mathrm{diag}(e^{j(\boldsymbol{\phi}-2\pi f_{s}\boldsymbol{\tau})})\boldsymbol{\beta}\big\rVert_{2}^{2}\Big] \\
=&\mathbb{E}\Big[\boldsymbol{\beta}^{\mathrm{H}}\mathrm{diag}(e^{j(\boldsymbol{\phi}-2\pi f_{s}\boldsymbol{\tau})})\big(\mathbf{R}_{\theta}\odot\mathbf{A}_{t}(\boldsymbol{\theta})\big)^{\mathrm{H}}\big(\mathbf{R}_{\theta}\odot\mathbf{A}_{t}(\boldsymbol{\theta})\big) \nonumber \\ &\,\, \times \mathrm{diag}(e^{j(\boldsymbol{\phi}-2\pi f_{s}\boldsymbol{\tau})})\boldsymbol{\beta}\Big] \\
=&\mathrm{tr}\Big(\mathbb{E}\Big[\big(\mathbf{R}_{\theta}\odot\mathbf{A}_{t}(\boldsymbol{\theta})\big)^{\mathrm{H}}\big(\mathbf{R}_{\theta}\odot\mathbf{A}_{t}(\boldsymbol{\theta})\big)\Big]\mathrm{diag}(e^{j(\boldsymbol{\phi}-2\pi f_{s}\boldsymbol{\tau})})\notag\\
&\,\, \times \mathbb{E}\big[\boldsymbol{\beta}\boldsymbol{\beta}^{\mathrm{H}}\big]\mathrm{diag}(-e^{j(\boldsymbol{\phi}-2\pi f_{s}\boldsymbol{\tau})})\Big). \label{first_term}
\end{align}
Since $\beta_{l}\sim\text{Unif}[0,\beta_{\mathrm{max}}]$, we get $\mathbb{E}\big[\boldsymbol{\beta}\boldsymbol{\beta}^{\mathrm{H}}\big]=\frac{\beta_{\mathrm{max}}^{2}}{12}\mathbf{I}_{L}$.
Also, using the definition of $\mathbf{R}_{\theta}$ in \eqref{R_theta} and the fact that $\Delta\theta_{u}\sim\text{Unif}\big[-\frac{\pi}{2^{Q_{\theta}}},\frac{\pi}{2^{Q_{\theta}}}\big]$, $(u,v)$-th element of $\mathbb{E}\big[\big(\mathbf{R}_{\theta}\odot\mathbf{A}_{t}(\boldsymbol{\theta})\big)^{\mathrm{H}}\big(\mathbf{R}_{\theta}\odot\mathbf{A}_{t}(\boldsymbol{\theta})\big)\big]$ is computed as
\begin{align}
    &\frac{4\pi^{2} d^{2}}{\lambda^{2}}\mathbb{E}\big[(\cos{\theta_{u}}\Delta\theta_{u}(\mathbf{n}_{N_{t}}\odot \mathbf{a}_{t}(\theta_{u})))^{\mathrm{H}}\notag\\
    &(\cos{\theta_{v}}\Delta\theta_{v}(\mathbf{n}_{N_{t}}\odot \mathbf{a}_{t}(\theta_{v})))\big]\notag\\
    &=\frac{4\pi^{2} d^{2}}{\lambda^{2}}\mathbb{E}[\cos{\theta_{u}}\cos{\theta_{v}}]\mathbb{E}[\Delta\theta_{u}\Delta\theta_{v}]\notag\\
    &\quad \times (\mathbf{n}_{N_{t}}\odot \mathbf{a}_{t}(\theta_{u}))^{\mathrm{H}}(\mathbf{n}_{N_{t}}\odot \mathbf{a}_{t}(\theta_{v}))  \\
    &\approx\delta_{u,v}\frac{4\pi^{2} d^{2}}{\lambda^{2}}\mathbb{E}\big[\cos^{2}{\theta_{u}}\big]\mathbb{E}\big[(\Delta\theta_{u})^{2}\big]\lVert\mathbf{n}_{N_{t}}\odot \mathbf{a}_{t}(\theta_{v})\rVert_{2}^{2}  \\
    &=\frac{\delta_{u,v}\pi^{4}d^{2}N_{t}(N_{t}-1)}{3\cdot 2^{2Q_{\theta}}\lambda^{2}}. \label{exp}
\end{align}
By plugging \eqref{exp} to \eqref{first_term}, we have
\begin{equation}
\mathbb{E}\big[\lVert(\nabla_{\boldsymbol{\theta}}\mathbf{h}[s])\Delta\boldsymbol{\theta}\rVert_{2}^{2}\big]\approx\frac{\pi^{4}d^{2}LN_{t}(N_{t}-1)\beta_{\mathrm{max}}^{2}}{36\cdot 2^{2Q_{\theta}}\lambda^{2}}.
\end{equation}
Also, we get the angle quantization term $C_{\theta}$ as
\begin{align}
C_{\theta}=&\sum_{s=1}^{N_{f}}\mathbb{E}\big[\lVert(\nabla_{\boldsymbol{\theta}}\mathbf{h}[s])\Delta\boldsymbol{\theta}\rVert_{2}^{2}\big] \nonumber \\
    \approx&\frac{\pi^{4}d^{2}LN_{f}N_{t}(N_{t}-1)\beta_{\mathrm{max}}^{2}}{36\cdot 2^{2Q_{\theta}}\lambda^{2}}.
\end{align}
Similarly, $C_{\tau}$, $C_{\beta}$, and $C_{\phi}$ can be expressed as functions of $Q_{\tau}$, $Q_{\beta}$, and $Q_{\phi}$, respectively.

\bibliographystyle{IEEEtran}

\begin{thebibliography}{1}

\bibitem{WRC23}
``{World RadioCommun. Conference 2023 (WRC-23), Available:
  https://www.itu.int/dms$\_$pub/itu-r/opb/act/R-ACT-WRC.15-2023-PDF-E.pdf},''
  2023.

\bibitem{rappaport2013millimeter}
Rappaport, T.~S \emph{et~al.}, ``Millimeter wave mobile communications for 5{G}
  cellular: {I}t will work!'' \emph{IEEE Access}, vol.~1, pp. 335--349, 2013.

\bibitem{ji2018ultra}
H.~Ji, S.~Park, J.~Yeo, Y.~Kim, J.~Lee, and B.~Shim, ``Ultra-reliable and
  low-latency communications in 5{G} downlink: {P}hysical layer aspects,''
  \emph{IEEE Wireless Commun.}, vol.~25, no.~3, pp. 124--130, 2018.

\bibitem{ahn2022towards}
Y.~Ahn \emph{et~al.}, ``Towards intelligent millimeter and terahertz
  communications for 6{G}: Computer vision-aided beamforming,'' \emph{IEEE
  Wireless Commun.}, vol.~30, no.~5, pp. 179--186, 2023.

\bibitem{kim2022parametric}
J.~Kim, Y.~Ahn, S.~Kim, and B.~Shim, ``Parametric sparse channel estimation
  using long short-term memory for {mmWave} massive {MIMO} systems,'' in
  \emph{Proc. IEEE Int. Conf. Commun. (ICC)}, 2022, pp. 1397--1402.

\bibitem{flordelis2018massive}
J.~Flordelis, F.~Rusek, F.~Tufvesson, E.~G. Larsson, and O.~Edfors, ``Massive
  {MIMO} performance—{TDD} versus {FDD}: {W}hat do measurements say?''
  \emph{IEEE Trans. Wireless Commun.}, vol.~17, no.~4, pp. 2247--2261, 2018.

\bibitem{guo2022overview}
J.~Guo, C.-K. Wen, S.~Jin, and G.~Y. Li, ``Overview of deep learning-based
  {CSI} feedback in massive {MIMO} systems,'' \emph{IEEE Trans. Commun.},
  vol.~70, no.~12, pp. 8017--8045, 2022.

\bibitem{lee2015antenna}
B.~Lee, J.~Choi, J.-Y. Seol, D.~J. Love, and B.~Shim, ``Antenna grouping based
  feedback compression for {FDD}-based massive {MIMO} systems,'' \emph{IEEE
  Trans. Commun.}, vol.~63, no.~9, pp. 3261--3274, 2015.

\bibitem{3gpp.38.843}
3GPP, ``{Study on Artificial Intelligence (AI)/Machine Learning (ML) for NR air
  interface },'' {3rd Generation Partnership Project (3GPP)}, Technical Report
  (TR) 38.843, 2023, version 1.1.0.

\bibitem{guo2022ai}
J.~Guo, C.-K. Wen, S.~Jin, and X.~Li, ``{AI} for {CSI} feedback enhancement in
  5{G}-advanced,'' \emph{IEEE Wireless Commun.}, vol.~31, no.~3, pp. 169--176,
  2022.

\bibitem{9838669}
P.~Madadi, J.~Jeon, J.~Cho, C.~Lo, J.~Lee, and J.~Zhang, ``Polar{D}ense{N}et: A
  deep learning model for {CSI} feedback in {MIMO} systems,'' in \emph{Proc.
  IEEE Int. Conf. Commun. (ICC)}, 2022, pp. 1294--1299.

\bibitem{wen2018deep}
C.-K. Wen, W.-T. Shih, and S.~Jin, ``Deep learning for massive {MIMO CSI}
  feedback,'' \emph{IEEE Wireless Commun. Lett.}, vol.~7, no.~5, pp. 748--751,
  2018.

\bibitem{liu2020efficient}
Z.~Liu, L.~Zhang, and Z.~Ding, ``An efficient deep learning framework for low
  rate massive {MIMO} {CSI} reporting,'' \emph{IEEE Trans. Commun.}, vol.~68,
  no.~8, pp. 4761--4772, 2020.

\bibitem{3gpp.36.211}
3GPP, ``{Evolved Universal Terrestrial Radio Access (E-UTRA); Physical channels
  and modulation},'' {3rd Generation Partnership Project (3GPP)}, Technical
  Specification (TS) 36.211, 2016, version 13.1.0.

\bibitem{vaswani2017attention}
A.~Vaswani \emph{et~al.}, ``Attention is all you need,'' \emph{Adv. Neural Inf.
  Process. Syst.}, vol.~30, 2017.

\bibitem{han2020deep}
Y.~Han, M.~Li, S.~Jin, C.-K. Wen, and X.~Ma, ``Deep learning-based {FDD}
  non-stationary massive {MIMO} downlink channel reconstruction,'' \emph{IEEE
  J. Sel. Areas Commun.}, vol.~38, no.~9, pp. 1980--1993, 2020.

\bibitem{love2008overview}
D.~J. Love, R.~W. Heath, V.~K. Lau, D.~Gesbert, B.~D. Rao, and M.~Andrews, ``An
  overview of limited feedback in wireless communication systems,'' \emph{IEEE
  J. Sel. Areas Commun.}, vol.~26, no.~8, pp. 1341--1365, 2008.

\bibitem{3gpp.36.212}
3GPP, ``{Multiplexing and channel coding},'' {3rd Generation Partnership
  Project (3GPP)}, Technical Specification (TS) 36.212, 2011, version 10.2.0.

\bibitem{kim2020downlink}
S.~Kim, J.~W. Choi, and B.~Shim, ``Downlink pilot precoding and compressed
  channel feedback for {FDD}-based cell-free systems,'' \emph{IEEE Trans.
  Wireless Commun}, vol.~19, no.~6, pp. 3658--3672, 2020.

\bibitem{lin2021deep}
B.~Lin, F.~Gao, S.~Zhang, T.~Zhou, and A.~Alkhateeb, ``Deep learning-based
  antenna selection and {CSI} extrapolation in massive {MIMO} systems,''
  \emph{IEEE Trans. Wireless Commun.}, vol.~20, no.~11, pp. 7669--7681, 2021.

\bibitem{zhang2021deep}
S.~Zhang, Y.~Liu, F.~Gao, C.~Xing, J.~An, and O.~A. Dobre, ``Deep learning
  based channel extrapolation for large-scale antenna systems: {O}pportunities,
  challenges and solutions,'' \emph{IEEE Wireless Commun.}, vol.~28, no.~6, pp.
  160--167, 2021.

\bibitem{wang2018deep}
T.~Wang, C.-K. Wen, S.~Jin, and G.~Y. Li, ``Deep learning-based {CSI} feedback
  approach for time-varying massive {MIMO} channels,'' \emph{IEEE Wireless
  Commun. Lett.}, vol.~8, no.~2, pp. 416--419, 2018.

\bibitem{kim2023towards}
W.~Kim, Y.~Ahn, J.~Kim, and B.~Shim, ``Towards deep learning-aided wireless
  channel estimation and channel state information feedback for 6{G},''
  \emph{J. Commun. Net.}, vol.~25, no.~1, pp. 61--75, 2023.

\bibitem{shi2015convolutional}
X.~Shi, Z.~Chen, H.~Wang, D.-Y. Yeung, W.-K. Wong, and W.-C. Woo,
  ``Convolutional {LSTM} network: {A} machine learning approach for
  precipitation nowcasting,'' \emph{Adv. Neural Inf. Process. Syst.}, vol.~28,
  2015.

\bibitem{han2019tracking}
Y.~Han, Q.~Liu, C.-K. Wen, M.~Matthaiou, and X.~Ma, ``Tracking {FDD} massive
  {MIMO} downlink channels by exploiting delay and angular reciprocity,''
  \emph{IEEE J. Sel. Signal Process.}, vol.~13, no.~5, pp. 1062--1076, 2019.

\bibitem{3gpp.38.211}
3GPP, ``{5G; NR; Physical channels and modulation},'' {3rd Generation
  Partnership Project (3GPP)}, Technical Specification (TS) 38.211, 2020,
  version 16.2.0.

\bibitem{wu2016survey}
J.~Wu and P.~Fan, ``A survey on high mobility wireless communications:
  {C}hallenges, opportunities and solutions,'' \emph{IEEE Access}, vol.~4, pp.
  450--476, 2016.

\bibitem{3gpp.38.901}
3GPP, ``5{G}; study on channel model for frequencies from 0.5 to 100 {GHz},'' {3rd
  Generation Partnership Project (3GPP)}, Technical Report (TR) 38.901, 2020,
  version 16.1.0.

\bibitem{samimi20163}
M.~K. Samimi and T.~S. Rappaport, ``3-{D} millimeter-wave statistical channel
  model for 5{G} wireless system design,'' \emph{IEEE Trans. Microw. Theory
  Techn.}, vol.~64, no.~7, pp. 2207--2225, 2016.

\bibitem{daubechies2004iterative}
I.~Daubechies, M.~Defrise, and C.~De~Mol, ``An iterative thresholding algorithm
  for linear inverse problems with a sparsity constraint,'' \emph{Commun. Pure
  Appl. Math.}, vol.~57, no.~11, pp. 1413--1457, 2004.

\bibitem{cui2022transnet}
Y.~Cui, A.~Guo, and C.~Song, ``\text{TransNet: Full} attention network for
  {CSI} feedback in {FDD} massive {MIMO} system,'' \emph{IEEE Wireless Commun.
  Lett.}, vol.~11, no.~5, pp. 903--907, 2022.

\bibitem{wu2021channel}
C.~Wu, X.~Yi, Y.~Zhu, W.~Wang, L.~You, and X.~Gao, ``Channel prediction in
  high-mobility massive {MIMO}: From spatio-temporal autoregression to deep
  learning,'' \emph{IEEE J. Sel. Areas Commun.}, vol.~39, no.~7, pp.
  1915--1930, 2021.

\bibitem{mashhadi2021pruning}
M.~B. Mashhadi and D.~G{\"u}nd{\"u}z, ``Pruning the pilots: Deep learning-based
  pilot design and channel estimation for {MIMO}-{OFDM} systems,'' \emph{IEEE
  Trans. Wireless Commun.}, vol.~20, no.~10, pp. 6315--6328, 2021.

\bibitem{jiang2022accurate}
H.~Jiang, M.~Cui, D.~W. Ng, and L.~Dai, ``Accurate channel prediction based on
  transformer: Making mobility negligible,'' \emph{IEEE J. Sel. Areas Commun.},
  vol.~40, no.~9, pp. 2717--2732, 2022.

\bibitem{dutta2019case}
S.~Dutta, C.~N. Barati, D.~Ramirez, A.~Dhananjay, J.~F. Buckwalter, and
  S.~Rangan, ``A case for digital beamforming at {mmWave},'' \emph{IEEE Trans.
  Wireless Commun.}, vol.~19, no.~2, pp. 756--770, 2019.

\end{thebibliography}

\end{document}